\documentclass[12pt,a4paper]{article}

\usepackage{amsfonts}
\usepackage{amsmath}
\usepackage{amssymb}
\usepackage{color}
\usepackage{diagrams}

\let\ssection=\section
\renewcommand{\section}{\setcounter{equation}{0}\ssection}

\newtheorem{definition}{Definition}[section]
\newtheorem{theorem}{Theorem}[section]

\newtheorem{lemma}[theorem]{Lemma}

\newenvironment{proof}[1][Proof]{\noindent\textbf{#1.} }{\ \rule{0.5em}{0.5em}}

\newcommand\eq[1]{(\ref{#1})}

\newcommand\ie{\emph{i.e.},\;}

\newcommand\Set{{\bf Sets}}

\renewcommand\a{\alpha}
\renewcommand\b{\beta}

\newcommand\das[1]{\delta(\hat{#1})}
\newcommand\dasto[2]{\delta(\hat{#2})_{#1}}

\newcommand{\dG}{\ps{\mkern1mu\raise2.5pt\hbox{$\scriptscriptstyle|$}
        {\mkern-7mu\rm O}}}                     
\newcommand{\dOU}{\ps{\mkern1mu\raise2.5pt\hbox{$\scriptscriptstyle|$}
        {\mkern-7mu\rm U}}}                     

\newcommand\de{\delta}

\newcommand{\ga}{\gamma}

\renewcommand\l{\lambda}

\newcommand{\tr}{{\rm tr}}

\newcommand\w{\mathfrak{w}}

\newcommand\BH{\mathcal{B(H)}}

\newcommand\De{\Delta}
\newcommand{\G}{\ps{O}}                         
\newcommand\Ga{\Gamma}

\newcommand{\Hi}{\mathcal{H}}
\newcommand\HG{{\rm Hyp}(\G)}
\renewcommand{\O}{\Omega}
\newcommand\Om{{\ps{\Omega}}}
\renewcommand{\P}{{\hat P}}
\newcommand\PH{\mathcal{P}(\mathcal{H})}
\newcommand\PV{\mathcal{P}(V)}
\newcommand{\PSig}{P_{{\rm cl}}\Sig}      

\newcommand\R{\mathcal{R}}

\newcommand\Sig{{\ps{\Sigma}}}
\newcommand{\SR}{\ps{{\mathbb{R}}^\succeq}}

\newcommand\UI{{\ps{[0,1]}^\succeq}}
\newcommand\Si{\Sigma}
\newcommand\TO{\mathbb{T}}
\newcommand\TOO{{}^{\rm org}\ps{\TO}}
\newcommand\VH{{\cal V}(\Hi)}

\newcommand\bra[1]{\langle #1|\,}
\newcommand\down[1]{\downarrow\!\!#1}
\newcommand\ddown[1]{_{|\downarrow\! #1}}

\newcommand\ket[1]{\,|#1\rangle}
\newcommand\ketbra[1]{\ket{#1}\bra{#1}}
\newcommand\la{\langle}
\newcommand\map{\rightarrow}
\newcommand\mathR{\mathbb{R}}

\newcommand\ra{\rangle}

\newcommand\LeftDB{[\mkern-3mu[}
\newcommand\RightDB{]\mkern-3mu]}
\newcommand\lr[2]{{\la#1,#2\ra}}

\newcommand\V[1]{{\cal V}(\Hi_{#1})}

\newcommand\Ob[1]{{\rm Ob(#1)}}
\newcommand\Sub[1]{{\rm Sub}(#1)}              
\newcommand\Subcl[1]{{\rm Sub}_{{\rm cl}}(#1)} 

\newcommand\Val[1]{\LeftDB\,#1\,\RightDB}
\newcommand\TVal[2]{\nu\big(#1;#2\big)}         

\newcommand\SAin[1]{\mbox{``}A\,\varepsilon\,#1\mbox{''}}
\newcommand\Ain[1]{A\,\varepsilon\,#1}

\newcommand\SetH[1]{\Set^{{\V{#1}}^{\rm op}}}  
\newcommand\SetC[1]{\Set^{{#1}^{\rm op}}}

\newcommand\Sh[1]{{\rm Sh}(#1)}

\newcommand\ShL[1]{{\rm Sh}(#1_L)}

\newcommand\sh[1]{\breve{\ps{#1}}}
\newcommand\ps[1]{\protect\underline{#1}}
\newcommand\eps[1]{\underline{\underline{#1}}}

\newtheorem{proposition}[theorem]{Proposition}

\definecolor{darkgreen}{rgb}{0,.66,0}

\newcommand\bbN{\mathbb{N}}

\begin{document}

\title{\textbf{Classical and Quantum Probabilities as Truth Values}}
\author{Andreas~D\"oring\footnote{andreas.doering@comlab.ox.ac.uk}\;  and
Chris~J.~Isham\footnote{c.isham@imperial.ac.uk}}
\date{February, 2011}
\maketitle

\begin{abstract}
We show how probabilities can be treated as truth values in suitable sheaf topoi. The scheme developed in this paper is very general and applies to both classical and quantum physics. On the quantum side, the results are a natural extension of our existing work on a topos approach to quantum theory. Earlier results on the representation of arbitrary quantum states are complemented with a purely logical perspective. 
\end{abstract}
\medskip
\begin{flushright}
``\textit{Fate laughs at probabilities.}''\\
E.G.~Bulwer-Lytton, from \textit{Eugene Aram} (1832)
\end{flushright}
\section{Introduction}\label{Sec:Intro}
In a long series of papers, we and our collaborators have shown
how quantum theory can be re-expressed
as a type of `classical physics' in the topos  of
presheaves (\ie set-valued contravariant functors) on the
partially-ordered set all commutative von Neumann
sub-algebras of the algebra of all bounded
operators on the quantum-theory Hilbert space $\Hi$
\cite{IB98,IB99,IB00,IB02,DI(1),DI(2),DI(3),DI(4),DI(Coecke)08}.
These ideas have been further developed by Caspers, Heunen, Landsman, Spitters, and Wolters in a way
that emphasises the internal structure of the topos
\cite{N1,N2,N3,N4,N5,N6}. Flori has presented a topos formulation of consistent histories in \cite{Flo09}.

The reformulation of quantum theory presented in these articles is
a radical departure from the usual Hilbert space formalism. All
aspects of quantum theory---states and state space, physical
quantities, the Born rule, etc.---find a new mathematical
representation, which also provides the possibility of a novel
conceptual understanding. In particular, an observer-independent,
non-instrumentalist interpretation becomes possible. For this
reason, we call the new formalism `neo-realist'.

Of course, many open questions remain. This article deals with  aspects of probability as it shows up in both
classical and quantum physics. As we will see, in both cases the usual
probabilistic description can be absorbed into
the logical framework supplied by topos theory.\footnote{As a side
remark, we do not see quantum theory fundamentally as some
kind of generalised probability theory. Such a viewpoint is almost
invariably based on an operational view of physics and, worse,
usually comes with a very unclear ontology of both probabilities
themselves and the objects or processes to which they apply.}

The interpretation of probability theory has been discussed
endlessly after the Renaissance endorsed it as a respectable
subject for study. In longevity, the subject shares the peristalithic
nature of debates about the conceptual meaning of quantum theory.
Every physicist even mildly interested in philosophical questions
will have heard about the range of different, incompatible
viewpoints about  probability: frequentist vs.~Bayesian
vs.~propensity; objective vs.~subjective probabilities; classical
vs.~quantum probabilities; epistemic (lack-of-knowledge) vs.~irreducible probabilities; and a bewildering range of combinations of those.

We cannot hope to solve this debate here, but some
useful remarks cane made from the viewpoint of physics. While most
scientists lean naturally towards the relative-frequency/frequentist view on probability, this interpretation is
limited because it cannot be used to assign a probability to the
outcome of a single experiment.\footnote{It is always interesting to
reflect on what a weather forecaster \emph{really} means when he
or she says "There is 80\% chance of snow tomorrow".} By definition, the frequentist interpretation requires a large ensemble of similar systems on which an experiment is performed, or a large
number of repetitions of the experiment on a single system. 

A\ particular challenge is posed by those physical situations in which a frequentist
interpretation cannot apply, even in principle. For example, if the
whole universe is regarded as a single entity, as in cosmology,
then clearly there are no multiple copies of the system. Moreover,
the instrumentalist concept of an `experiment' performed on
the entire universe is  meaningless, since there is no
external observer or agent who could perform such an experiment. This renders problematic both quantum
cosmology and stochastic classical cosmology unless probabilities
can be understood in non-instrumentalist terms.  This argument
applies also to subsystems of the universe provided they are
sufficiently large and unique to make impossible the preparation
of an ensemble of similar systems, or repetitions of an experiment
on the same system.

Of course, in most of science there \emph{is} a valid instrumentalist view in which the world is divided into a system, or
ensemble of systems, and an observer. The system, or ensemble,
shows probabilistic behaviour when an observer performs
experiments on it. In the ensuing two-level ontology the system
and the observer have very different conceptual status.
Frequentist views of probability typically lead to such a dualism.
A Bayesian view, in which probabilities are primarily states
of knowledge or evidence, also presupposes a divide between system and
observer and is based on an operational way of thinking about
physical systems.

Such an operational view does not readily extend  to
quantum cosmology.\footnote{We are aware of many-worlds approaches
and the attempts to (re)define probabilities in such a formalism. However, we are not very enthusiastic about these schemes.}
Hence, it is desirable to have a (more) realist formulation of
quantum theory---or, potentially, more general theories---that could
apply meaningfully to the whole universe. This desire to avoid the
two-level ontology of operational/instrumentalist approaches is
one of the motivations for the topos approach to the
formulation of physical theories.

In our previous work  it was shown in detail how pure quantum
states and propositions are represented in the topos formalism and
how truth values can be assigned to all propositions, without any
reference to observers, measurements or other instrumentalist
concepts. In fact, for pure quantum states, probabilities are
\emph{replaced} by truth values which are given by the structure
of the topos itself. For the specific presheaf topoi  used in our
reformulation of quantum theory, a truth value is a lower set  in
the set $\VH$ (which is partially ordered under inclusion) of
commutative subalgebras of the algebra, $\BH$, of all bounded
operators on $\Hi$.  We only consider non-trivial, commutative von
Neumann subalgebras $V\subset\BH$ that contain the identity
operator. Each $V\in\VH$ can be seen as providing a classical
perspective on the quantum system, with smaller commutative
subalgebras giving a more `coarse-grained' perspective than bigger
ones.

A truth-value is therefore a collection of classical perspectives
from which a given proposition is true. The fact that a truth
value is a \emph{lower} set in $\VH$ expresses the idea that once
a proposition is `true from the classical perspective $V$', upon
coarse-graining to smaller subalgebras $V'\subset V$, the
proposition should stay true. The elements\footnote{Strictly speaking, $\VH$\ is a category whose \emph{objects} are the commutative sub-algebras $V$. It is therefore more accurate to write $V\in\Ob{\VH}$, rather than $V\in\VH$, and this we did in our earlier papers. However, here, for the sake of simplicity, we use the latter notation.} $V\in\VH$ are also
called \emph{contexts} or, by mathematicians, \emph{stages of
truth}. It is easy to see that there are uncountably many possible truth values if ${\rm dim}(\Hi)>1$ (compare this with the uncountably
many probabilities in the interval $[0,1]$.)

The treatment of pure states in the topos approach makes it
unnecessary to speak of probabilities in any fundamental way. However,
at first sight mixed states cannot be treated similarly. Since we
aim at a formalism that can be interpreted in a (neo-)realist way,
it seems most appropriate to regard probabilities as objective.
More specifically, we lean towards an interpretation of
probabilities as propensities. Since the probabilities associated
with pure quantum states are absorbed into the logical structure
given by the topos, we aspire to find a `logical reformulation'
for the probabilities associated with mixed states as well. As we
will see, this necessitates an extension of the topos used so far for
quantum theory. Probabilities are thereby built into the
mathematical structures in an intrinsic manner. They are tied up with the
internal logic of the topos and do not show up as external
entities to be introduced when speaking about experiments.

We finally remark that all constructions shown here work for
arbitrary von Neumann algebras\footnote{For some results, the
algebra must not have a type $I_2$-summand.} and arbitrary states,
normal or non-normal. The more general proofs need no extra
effort, though interpretational subtleties relating to non-normal
states may arise. For simplicity and clarity of presentation, we
use here only $\BH$ as the algebra of physical quantities of a
quantum system and pure or mixed (\ie normal) states.

\section{The Topos Approach and Mixed States}
\subsection{Some basic definitions}
There are several articles giving an introduction to the topos
approach to quantum theory \cite{Doe07b,DI(Coecke)08,Doe10} and only a few ingredients are sketched here. We
assume some familiarity with basic aspects of category and topos
theory and of functional analysis.

A key feature of the topos approach
is the existence  for each quantum  system of an object that is
functionally analogous to the state space of a classical system.
This `quantum state space', $\Sig$, is a presheaf (\ie a
set-valued, contravariant functor) on the poset $\VH$ of abelian
subalgebras of $\BH$, the algebra of physical quantities (or observables) of the quantum system. The collection of all such functors is a topos, denoted $\SetC\VH$. The poset $\VH$ and object $\Sig$ are
known respectively as the \emph{context category} and
\emph{spectral presheaf}.

\begin{definition}
The \emph{spectral presheaf}, $\Sig$, is defined over the context category, $\VH$, as follows:
\begin{enumerate}
        \item [(i)] On objects  $V\in\VH$,  $\Sig_V$
        is the Gel'fand spectrum of the commutative von Neumann algebra $V$.
        \item [(ii)] On morphisms
        $i_{V'V}:V'\map V$ (\ie $V'\subseteq V$), the presheaf functions
$\Sig(i_{V'V}):\Sig_V\map \Sig_{V'}$ are defined as
        \begin{align}
             \Sig(i_{V'V})(\l):= \l|_{V'}
        \end{align}
for $\l\in\Sig_V$. Here $\l|_{V'}$ denotes the restriction to $V'$
of the spectral element $\l\in\Sig_V$.
\end{enumerate}
\end{definition}

A sub-object (\ie sub-presheaf) $\ps S$ is
said to be \emph{clopen} if $\ps S_V$ is a clopen set in $\Sig_V$
for all $V\in\VH$. The clopen sub-objects of $\Sig$ form a
complete Heyting algebra (Theorem 2.5 in \cite{DI(2)}), denoted 
$\Subcl{\Sig}$.

Propositions in physics 
are usually of the form $\SAin\De$, which in classical physics would read
``the physical quantity $A$ has a value, and that value lies in
the (Borel) set $\De\subseteq\mathR$ of real numbers''. Using the spectral
theorem for self-adjoint operators, a proposition $\SAin\De$ is
represented in quantum theory by the spectral projector $\hat E[\Ain\De]$. It was
shown in \cite{DI(2)} that there is a map
$\ps\de:\PH\map\Subcl{\Sig}$, called `\emph{daseinisation} of
projection operators', which sends each projection operator $\P$
to a clopen sub-object $\ps{\de(\P)}$. In this way, a proposition
$\SAin\De$ is represented by the clopen sub-object $\ps{\de(\hat
E[\Ain\De])}$ of $\Sig$. This is analogous to classical physics
where propositions are represented by (measurable) subsets of the
classical state space.

The sub-object $\ps{\de(\P)}$ is constructed in a two-step process.
First, one defines, for all $V\in\VH$,
\begin{equation}
     {\de(\P)}_V:=\bigwedge\{\hat Q\in\PV \mid \hat Q\succeq\P\},
     \label{Def:de(P)}
\end{equation}
which  gives the `best' approximation to $\P$ from above by a projector in the lattice, $\PV$, of projection operators in the  context $V$. The association $V\mapsto\de(\P)_V$ is a global element, denoted $\de(\P)$,
of the \emph{outer} presheaf $\G$ \cite{IB98,IB00}:

\begin{definition}
\label{DefG} The \emph{outer presheaf}, $\G$,  is defined as follows
:
\begin{enumerate}
\item[(i)] On objects $V\in\VH$,   $\G_V:=\PV$.

\item[(ii)] On morphisms $i_{V^{\prime}V}:V^{\prime
}\subseteq V$, the mapping $\G(i_{V^{\prime} V}):\G_V
\map\G_{V^{\prime}}$ is
$\G(i_{V^{\prime}V})(\P):=\dasto{V^\prime}{P}$ for all $\P\in\PV$.
\end{enumerate}
\end{definition}

The second step uses the existence of a monic arrow,
$\iota:\G\map\PSig$, from $\G$ to the \emph{clopen power object}, $\PSig$, of $\Sig$ \cite{DI(2)}. This sub-object, $\PSig$, of the  power object $P\Sig$ has the property that its global elements are \emph{clopen} sub-objects of $\Sig$, whereas the global elements of $P\Sig$ are arbitrary sub-objects \cite{DI(2)}. The construction of $\iota$ exploits the fact that for any commutative von
Neumann algebra $V$, there is an isomorphism of complete Boolean
algebras
\begin{eqnarray}
   \a_V:\PV=\G_V &\map& \mathcal{C}l(\Sig_V),\label{Def:aV}\\
     \P\   \ &\mapsto& \{\l\in\Sig_V \mid \l(\P)=1\}\nonumber
\end{eqnarray}
between the projections in $V$ and the clopen subsets of the
Gel'fand spectrum, $\Sig_V$, of $V$. Given a projection
$\P\in\PV$, we define 
\begin{equation}
        S_\P:=\alpha_V(\P),
\end{equation}
and given a clopen subset $S\subseteq\Sig_V$, we define
\begin{equation}
        \P_S:=\alpha_V^{-1}(S).
\end{equation} 
Thus locally, \ie
in each context $V$, we can switch between clopen subsets and
projections.

Applying this to the family $\de(\P)_V$, $V\in\VH$, of projectors, gives a family  $S_{\de(\P)_V}\subseteq\Sig_V$, $V\in\VH$, of clopen subsets. This family of clopen subsets forms a
clopen sub-object, denoted $\ps{\de(\P)}$, of $\Sig$. Equivalently, the global
element $\de(\P):\ps1\map\G$ defined by \eq{Def:de(P)} gives rise via $\iota:\G\map\PSig$ to a global
element $\ps1\map\G\map\PSig$ of $\PSig$, and hence to the clopen sub-object $\ps{\de(\P)}$.

Given a (normalised) vector state $\ket\psi$, we now form
\begin{equation}
     \ps\w^{\ket\psi}:=\ps{\de(\ket\psi\bra\psi)}\in\Subcl{\Sig}.
\end{equation}
This clopen sub-object is the topos representative of the pure
state $\ket\psi$: we call $\ps\w^{\ket\psi}$ the
\emph{pseudo-state} associated with $\ket\psi$.

To each  vector state $\ket\psi\in\Hi$ and physical proposition
$\SAin\De$ there corresponds the `topos truth value'
$\TVal{\Ain\De}{\ket\psi}\in\Ga\Om$ defined at each stage
$V\in\VH$ as,
\begin{equation}
\TVal{\Ain\De}{\ket\psi}(V):=\{V'\subseteq V\mid\bra\psi
                \delta(\hat E[\Ain\De])_{V'}\ket\psi=1\}.
                \label{TTVpsi}
\end{equation}
We also write
\begin{equation}
\TVal{\P}{\ket\psi}(V):=\{V'\subseteq V\mid\bra\psi
                \delta(\P)_{V'}\ket\psi=1\}
                \label{TTVPpsi}
\end{equation}
for any projection operator $\P$ on $\Hi$. 

This truth value, which is a sieve on $V$, can be
understood as the `degree' to which the sub-object
$\ps\w^{\ket\psi}$ that represents the pure state is contained in
the sub-object $\ps{\delta(\hat E[\Ain\De])}$ that represents the
proposition. This degree is simply the collection of all those
contexts $V\in\VH$ such that $\ps\w^{\ket\psi}_V\subseteq\ps{\delta(\hat E[\Ain\De])}_V$.

\subsection{The problem of mixed states}
A very interesting (and important) question is how mixed states
should be treated in this topos formalism. In its basic form, a
mixed state is just a collection of vectors and `weights',
$\{\ket{\psi_1},\ket{\psi_2},\ldots,\ket{\psi_N}; r_1,r_2,\ldots,
r_N\}$, where $\sum_{i=1}^Nr_i=1$ (we allow $N=\infty$). In
standard quantum theory, the assumption that the quantum probabilities are stochastically
independent from the `weights' $r_1,r_2,\ldots, r_N$, leads via elementary probability arguments
to the
familiar definition of the associated density matrix as
$\rho:=\sum_{i=1}^Nr_i\ket{\psi_i}\bra{\psi_i}$ and hence to the
familiar expressions in which, for example, $\tr(\rho\hat A)$
replaces $\bra\psi\hat A\ket\psi$.

The important challenge is to find an analogue
for $\rho$ of the truth value in \eq{TTVpsi} for
vector states $\ket\psi$, and in such a way that density matrices are separated by this expression. Of course, the vector $\ket\psi$ in \eq{TTVpsi}
can be replaced with  $\rho$ to give the sieve on $V$
\begin{equation}
\TVal{\Ain\De}{\rho}(V):=\{V'\subseteq V\mid{\rm tr}\big(
                \rho\,\delta(\hat E[\Ain\De])_{V'}\big)=1\},
                \label{TTVrho}
\end{equation}
but this is not adequate for our needs as it does \emph{not} separate density matrices. 

For example,
let $\Hi$ be a three-dimensional Hilbert space with a chosen
basis, and let
$\rho=\operatorname{diag}(\frac{1}{2},\frac{1}{2},0)$ and
$\tilde\rho=\operatorname{diag}(\frac{3}{4},\frac{1}{4},0)$ be two
density matrices that are diagonal with respect to this basis.
Then ${\rm tr}(\rho\P)=1$ if and only if ${\rm
tr}(\tilde\rho\P)=1$ if and only if $\P\succeq\P_1+\P_2$, where
$\P_1,\P_2$ are the projections onto the rays determined by the
first two basis vectors. In other words, $\rho$ and $\tilde\rho$
have the same support, namely $\P_1+\P_2$, and the topos truth
value $\TVal{\Ain\De}{\rho}\in\Ga\Om$ (given by the family
$\TVal{\Ain\De}{\rho}(V)$, ${V\in\VH}$, of sieves defined
above) depends only on the supports, not on the actual weights, in
$\rho$ and $\tilde\rho$.

As pointed out in \cite{IB00},  there exists a
one-parameter family of valuations defined for all 
$V\in\VH$ as
\begin{equation}
\TVal{\Ain\De}{\rho}^r(V):=\{V'\subseteq V\mid
\tr\big(\rho\,\delta(\hat E[\Ain\De])_{V'}\big)\geq r\}
                \label{TTVrho_r}
\end{equation}
where $r\in (0,1]$. (One could allow $r=0$, but clearly, the condition then is trivially fulfilled and the valuation $\nu(-;-)^0$ gives the truth value `totally true', represented by the maximal sieve at each stage, for all states $\rho$ and all propositions ``$\Ain\De$''.) The authors of \cite{IB00} could find no use
for these $r$-modified `truth' values. However,  it transpires
that this family of valuations \emph{does} separate density
matrices, which is very suggestive of how to proceed.

Indeed, the main result of this paper is to show how the
one-parameter family in \eq{TTVrho_r} can be regarded as a \emph{single} valuation
in a particular extension of the topos $\SetH{}$. By this means  we will achieve a completely topos-internal description of arbitrary
(normal) quantum states $\rho$.

At this point we remark that there is another topos perspective on
density matrices that at first sight appears to be very different
from the one above. One of us (AD) has shown how each density
matrix, $\rho$, gives rise to a `probability' (pre-)measure,
$\mu^\rho$, on $\Subcl\Sig$ \cite{ADmeas}. This function
$\mu^\rho:\Subcl{\Sig}\map\Ga\ps{[0,1]}^\succeq$ is defined at all
stages $V\in\VH$ by
\begin{equation}
        \mu^\rho(\ps{S})(V):={\rm tr}(\rho\,\P_{\ps{S}_V})
                                \label{Def:murho}
\end{equation}
where $\P_{\ps{S}_V}$ is  the projection operator in
${\cal P}(V)$ that corresponds to the component $\ps S_V$ of the clopen sub-object $\ps{S}$
at stage $V$. Here, $\UI$ denotes the presheaf of $[0,1]$-valued,
nowhere-increasing functions on the poset/category $\VH$. It can readily be checked that this family of measures separates density matrices.

Conversely, $\UI$-valued probability measures on the spectral presheaf $\Sig$
can be defined abstractly, with the clopen sub-objects playing the
role of measurable subsets. Provided the Hilbert space $\Hi$ has
at least dimension three,  from each such
measure, $\mu$, one can construct a unique quantum state
$\tilde\rho^\mu:\BH\map\mathbb{C}$, such that
$\mu^{\tilde\rho^\mu}=\mu$.\footnote{For the more general case of von
Neumann algebras treated in \cite{ADmeas}, the condition is that
the von Neumann algebra has no summand of type $I_2$.} The
existence of such measures is in accord with our general slogan
that ``Quantum physics is equivalent to classical physics in the
topos $\SetH{}$''.

The construction in \eq{Def:murho} can be
applied to a (normalised) vector state $\ket\psi$, to give
\begin{equation}
  \mu^{\ketbra\psi}(\ps{S})(V):=\bra\psi\P_{\ps{S}_V}\ket\psi
                                \label{Def:mupsi}
\end{equation}
at all stages $V$. In particular, for the sub-object
$\ps{S}_{\Ain\De}$ of $\Sig$ associated with the spectral
projector $\hat E(\Ain\De)$ we get
\begin{equation}
\mu^{\ketbra\psi}(\ps{S}_{\Ain\De})(V)=\bra\psi\de(\hat
E[\Ain\De])_V\ket\psi
        \label{muAinDelta}
\end{equation}

Thus, for each state $\ket\psi$, a proposition $\SAin\De$ is
associated with two, quite different, mathematical entities: 
\begin{enumerate}
\item[(i)]
the Heyting-algebra valued topos truth value
$\TVal{\Ain\De}{\ket\psi}\in\Ga\Om$ defined in \eq{TTVpsi} as
\begin{equation}
\TVal{\Ain\De}{\ket\psi}(V):=\{V'\subseteq V\mid\bra\psi
                \delta(\hat E[\Ain\De])_{V'}\ket\psi=1\}
                \label{TTVpsi2}
\end{equation}
for all $V\in\VH$; and 

\item[(ii)]the $\UI$-valued measure $\mu^{\ketbra\psi}$ on
$\Subcl\Sig$ defined in \eq{muAinDelta} as
\begin{equation}
\mu^{\ketbra\psi}(\ps{S}_{\Ain\De})(V):=\bra\psi\de(\hat
E[\Ain\De])_V\ket\psi
        \label{muAinDelta2}
\end{equation}
for all $V\in\VH$.
\end{enumerate}

The purpose of the present paper is to study the relation between
these two entities, and  thereby to see if there is a topos-logic
representation of the general, density-matrix measure, $\mu^\rho$
in \eq{Def:murho}. Thus the challenge is to relate the probability
measure $\mu^\rho:\Subcl{\Sig}\map\Ga\ps{[0,1]}^\succeq$ to a topos 
truth value and hence to relate probability to intuitionistic logic. 
It is clear that the definition in \eq{TTVrho} is not adequate as, 
unlike the measures $\mu^\rho$, it fails to separate density matrices.

As we shall see, the one-parameter family of valuations $r\mapsto
\TVal{\Ain\De}{\rho}^r\in\Ga\Om$ defined in \eq{TTVrho_r} plays a key role. However, the incorporation of this family into a topos framework requires
an extension of the original quantum topos $\SetH{}$. We will 
motivate this by first considering a topos perspective on classical
probability theory.

%
%

We conclude this section with a technical remark. Namely, instead
of talking about presheaves over the context category $\VH$, which
is a poset, we can  equivalently talk about sheaves if $\VH$ is
equipped with the (lower) Alexandroff topology in which the open
sets are the lower sets in $\VH$. Then, by a standard
result\footnote{If $\sh{A}$ denotes the sheaf associated with the
presheaf $\ps{A}$, then, on the open set $\down{V}$, we have
$\sh{A}(\down V)=\ps{A}_V$ for all $V\in\VH$. On an arbitrary lower
set $L\subseteq\VH$, $\sh{A}(L)$ is the (possibly empty) set of 
local sections of $\Sig$ over $L$.} we have
\begin{equation}                        \label{Eq_PresheafSheafTopos}
        \Sh{\VH_A}\simeq\SetC\VH,
\end{equation}
where $\VH_A$ denotes the poset $\VH$ equipped with the lower
Alexandroff topology. In what follows we will move freely between
the language of presheaves and that of sheaves---which of the two
is being used should be clear from the context. 

\section{A Topos Representation for Probabilities}\label{Sec_ToposRepOfProbabs}
The usual description of probabilities is by numbers in the interval $[0,1]$, equipped with the total order inherited from $\mathbb R$. In order to absorb probabilities into the logic of a topos, our goal now is to find a topos, $\tau$, such that
\begin{equation}
[0,1]\simeq\Ga\O^\tau                      \label{hsimeqOmtau}
\end{equation}
where $\O^\tau$ is the truth object in
$\tau$. A natural way of doing this is to look for a topological space, $X$, whose open sets correspond bijectively with the numbers between $0$ and $1$.

Since the global elements of the sub-object classifier $\O^\tau$ are the truth values available in the topos $\tau$, equation \eq{hsimeqOmtau} means that probabilities will correspond bijectively with truth values in the topos. 

A straightforward idea is to use suitable lower sets in the unit interval $[0,1]$. These lower sets are required to form a topology, so the collection of lower sets must be closed under arbitrary unions and finite intersections. The intervals of the form $(0,r)$, where $0\leq r\leq 1$, form a topology. Here $(0,0)=\emptyset$. The maximal element is $(0,1)$, so our topological space $X$ actually is the interval $(0,1)$ (and not $[0,1]$). When regarding $(0,1)$ as a topological space with the topology given by the sets of the form $(0,r)$, $0\leq r\leq 1$, we denote it as $(0,1)_L$. We define a bijection
\begin{align}                   \label{beta}
                        \beta:[0,1] &\map {\cal O}((0,1)_L)\\                   \nonumber
                        r &\mapsto (0,r).
\end{align}
Note also that intervals of the form $(0,r]$ would not work: they are not closed under arbitrary unions. Take for example all intervals $(0,r_i]$ such that $r_i<r_0$ for some fixed $r_0\in [0,1]$. Then $\bigcup_{i}(0,r_i]=(0,r_0)$.

It may seem odd that a probability $r\in[0,1]$ is represented by the (open) set $(0,r)$, which does not contain $r$, but this is not problematic. In the following, we will not interpret $(0,r)$ as the collection of all probabilities between $0$ and $r$, but as an open set, which corresponds to a truth value in a sheaf topos. This truth value, given by the structure of the topos, corresponds to the probability $r$.

Now, a key idea in topos theory is that any
topological space, $X$, should be replaced with the topos,
$\Sh{X}$, of sheaves over $X$. Furthermore, a standard result is
that there is an isomorphism of Heyting algebras.
\begin{equation}
                {\cal O}(X)\simeq \Ga\Om^{\Sh{X}}
\end{equation}
It becomes clear that the topos we are seeking is
$\ShL{(0,1)}$. We will simplify the notation $\Om^{\ShL{(0,1)}}$ to
just $\Om^{(0,1)}$; thus
\begin{equation}
        {\cal O}((0,1)_L)\simeq\Ga\Om^{(0,1)}.          \label{hsimeqGaO}
\end{equation}
We denote the isomorphism as $\sigma:{\cal O}((0,1)_L)\map\Ga\Om^{(0,1)}$. Its concrete form will be discussed below.

It is clear that, for all stages $(0,r)\in {\cal O}((0,1)_L)$, the component $\Om_r^{(0,1)}$ of the sheaf $\Om^{(0,1)}$ is given by
\begin{eqnarray}
                        \Om_{(0,r)}^{(0,1)}&=&\{(0,r')\mid0< r'\leq r\}\cup\emptyset\nonumber\\
      &=&\{(0,r')\mid 0\leq r'\leq r\}.
\end{eqnarray}
We remark that instead of describing stages as $(0,r)\in {\cal O}((0,1)_L)$, we can also think of $r\in [0,1]$ by the isomorphism \eq{beta}.

We now define a key map $\ell$ that takes a probability $p\in [0,1]$ into a global section of $\Om^{(0,1)}$, that is, a truth value in the topos. Roughly speaking, the idea is that $p$ is mapped to the open set $(0,p)$, which is a truth value in the sheaf topos. In fact, the map $\ell$ that we will define is nothing but the composition $\sigma\circ\beta$ of the set isomorphisms $\beta$ in \eq{beta} and $\sigma$ in \eq{hsimeqGaO}.

Concretely, we define for all $p\in [0,1]$ and all stages $(0,r)\in {\cal O}((0,1)_L)$ the following sieve on $(0,r)$:
\begin{eqnarray}
           \ell(p)_{(0,r)} &:=& \{(0,r')\in{\cal O}((0,1)_L) \mid p\geq r'\}\label{Def:ka(a)(r)}\\
     &=&\{(0,r') \mid r'=\min\{p,r\}\}
\end{eqnarray}
where, by definition, $(0,0)=\emptyset$. Thus we have
\begin{equation}
        \ell(p)_{(0,r)}=\left\{\begin{array}{ll}
                                \{(0,r')\in {\cal O}((0,1)_L) \mid r'\leq r\}=\Om_{(0,r)}^{(0,1)} &\mbox{ if $p\geq r$}\\
                                        \{(0,r')\in {\cal O}((0,1)_L) \mid r'\leq p\} &\mbox{ if $0<p < r$}\\
                                \emptyset&\mbox{ if $p=0$}
                              \end{array}
                         \right.
 \end{equation}
or simpler
\begin{equation}                        \label{ell}
        \ell(p)_r=\left\{\begin{array}{ll}
                                                                                        [0,r]=\Om_r^{(0,1)} &\mbox{ if $p\geq r$}\\
                              $[$0,p] &\mbox{ if $0<p<r$}\\
                            \emptyset&\mbox{ if $p=0.$}
                          \end{array}
                   \right.
\end{equation}
Here, we used the bijection $\{(0,r')\in{\cal O}((0,1)_L) \mid r'\leq r\}\simeq [0,r]$, etc., which of course is implied by equation \eq{beta}.

\section{The Construction of Truth Objects}
\subsection{The relation $\Ga(P\G)\simeq\Subcl\Sig$}
\label{Subsec:GaPO=Sub} The topos truth value \eq{TTVpsi} for a
vector state $\ket\psi$ can be viewed in two ways.  The first  employs the pseudo-state
$\ps\w^{\ket\psi}:=\ps{\de(\ketbra\psi)}\in\Subcl\Sig$ where, in the
notation of local set theory, we have\footnote{In local set
theory, to any pair, $\ps{S}_1,\ps{S}_2$ of sub-objects of a third object $\ps{A}$, there is
associated an element of $\Ga\Om$, denoted
$\Val{\ps{S}_1\subset\ps{S}_2}$, which measures the `extent' to
which it is true that $\ps{S}_1$ is  a sub-object of $\ps{S}_2$. In our case,  $\Val{\ps{S}_1\subset\ps{S}_2}(V):=\{V'\subseteq V\mid {\ps{S}_1}_{V'}\subseteq {\ps{S}_2}_{V'}\}$ for all $V\in\VH$.}
\begin{equation}
\TVal{\Ain\De}{\ket\psi}= \Val{\ps\w^{\ket\psi}\subseteq
\ps{\de(\hat E[\Ain\De])}}\in\Ga\Om\label{nupseudo}
\end{equation}
where the right hand side is well-defined since both
$\ps\w^{\ket\psi}$ and $\ps{\delta(\hat E[\Ain\De])}$ are clopen
sub-objects of the spectral presheaf $\Sig$.

Our goal  is to find a family of `truth objects',
$\ps{\TO}^{\ket\psi}$, $\ket\psi\in\Hi$, with the property that
the topos truth value in \eq{nupseudo} can be expressed
alternatively as
\begin{equation}
\TVal{\Ain\De}{\ket\psi}= \Val{\ps{\de(\hat
E[\Ain\De])}\in\ps{\TO}^{\ket\psi}}\label{deinT}
\end{equation}
We note that for \eq{deinT} to be meaningful,
$\ps{\TO}^{\ket\psi}$ must be a sub-object of the presheaf
$\PSig$.

In our earlier work on truth objects we  constructed the
quantities\footnote{The `org' is short for `original'.} $\TOO^{\ket\psi}$, $\ket\psi\in\Hi$, defined by
\begin{eqnarray}
\TOO^{\ket\psi}_V&:=&\{\hat\a\in\G_V\mid
       {\rm Prob}(\hat\a;\ket\psi)\}=1\nonumber\\
  &=&\{\hat\a\in\G_V\mid\bra\psi\hat\a\ket\psi\}=1\}
  \label{Def:TOpsi}
\end{eqnarray}
for all $V\in\VH$. It can readily be checked that \eq{Def:TOpsi}
defines a sub-object of $\G$. With the aid of the monic arrow
$\iota:\G\map\PSig$ this gives a sub-object of $\PSig$ to which \eq{deinT} applies.

For discussing daseinised sub-objects of $\Sig$ like
$\ps{\de({\hat E[\Ain\De])}}$ the simple definition in
\eq{Def:TOpsi} is sufficient. However, the situation changes if we
want to consider more general sub-objects of $\Sig$. Specifically,
we want to define $\ps\TO^{\ket\psi}$ in such a way that
\begin{equation}
\Val{\ps{S}\in\ps{\TO}^{\ket\psi}}\in\Ga\Om\label{SinT}
\end{equation}
is well-defined for \emph{any} clopen sub-object, $\ps{S}$, of
$\Sig$. This is necessary to fulfil our desire to relate topos
truth values with measures on $\Subcl\Sig$. Clopen sub-objects of
the form $\ps{\de({\hat E[\Ain\De])}}$ have very special
properties whereas our measures are defined on arbitrary clopen
sub-objects of $\Sig$. This necessitates a new definition of the
truth objects $\ps{\TO}^{\ket\psi}$, $\ket\psi\in\Hi$.

Evidently, \emph{clopen} sub-objects of the spectral presheaf
$\Sig$ are of particular interest as the analogues of measurable
subsets of a classical state space $\mathcal{S}$. More precisely,
as in \eq{Def:murho}, each quantum state $\rho$ determines
a ($\Ga\UI$-valued) `probability measure', $\mu^\rho$, on $\Sig$,
with the clopen sub-objects playing the role of measurable
subsets; conversely, each probability measure on
$\Subcl\Sig$ determines a unique quantum state. The
proof relies on Gleason's theorem which shows that a quantum state
is determined by the values it takes on projections. As we shall
see, clopen sub-objects have components which correspond to
projections, hence Gleason's theorem is applicable.

It is very desirable to be able to express all clopen sub-objects
of $\Sig$ in terms of projection operators as these are the
mathematical entities that have the most direct physical meaning
in quantum physics and they are also relatively easy to
manipulate. 

As mentioned earlier, there is  a monic arrow $\iota:\G\map\PSig$
 and, therefore, any global element of $\G$ leads to
a clopen sub-object of $\Sig$. However, $\iota$ is not surjective
and hence not all clopen sub-objects can be obtained in this way.
The remaining sub-objects can be recovered using the, so-called,
\emph{hyper-elements} of $\G$ that were introduced in
\cite{DI(2)}. This is an intermediate step towards realising our
main goal, which is to find some object $\ps{X}$ in the topos that
(i) can be defined purely in terms of projection operators; and
(ii) is such that\footnote{The power object $P\Sig$ is \emph{not}
the correct choice as $\Ga(P\Sig)\simeq\Sub\Sig$ and the latter
includes sub-objects of $\Sig$ that are not clopen; \ie
$\Subcl\Sig\subset\Sub\Sig$.} \begin{equation}
        \Ga\ps{X}\simeq\Subcl\Sig\label{GaX=Subcl}
\end{equation}

We start by recalling that a global element of $\G$ is a family of
elements $\hat\ga_V\in\G_V\simeq\PV$, $V\in\VH$, such that, for
all pairs $i_{V'V}:V'\subseteq V$ we have
\begin{equation}
        \G(i_{V'\,V})(\hat\ga_V)=\hat\ga_{V'}
\end{equation}
In other words
\begin{equation}
        \das{\ga_V}_{V'}=\hat\ga_{V'}.
\end{equation}

A \emph{hyper-element} of $\G$ is a generalisation of this concept. Specifically, a hyper-element is a family of elements  $\hat\ga_V\in\G_V$, $V\in\VH$, such
that, for all pairs $i_{V'V}:V'\subseteq V$ we have \cite{DI(2)}
\begin{equation}
        \G(i_{V'\,V})(\hat\ga_V)\preceq\hat\ga_{V'}
\end{equation}
In other words
\begin{equation}
        \de(\hat\ga_V)_{V'}\preceq\hat\ga_{V'}.
\end{equation}
We denote the set of all hyper-elements of $\G$ as $\HG$.

The importance of hyper-elements comes from the following
results which are proved in the Appendix.
\begin{enumerate}
\item There is a bijection 
\begin{eqnarray}
k:\HG&\map&\Subcl\Sig\label{Def:k}\\
k(\ga)_V&:=&\a_V(\hat\ga_V)=S_{\hat\ga_V}\nonumber
\end{eqnarray}
for all $\ga\in\HG$, $V\in\VH$ (see Proposition \ref{Theorem:HypO=Subcl}).

The inverse is
\begin{eqnarray}
j:\Subcl{\Sig}&\map&\HG\label{Def:j}\\
j(\ps{S})_V&:=&\a_V^{-1}(\ps{S}_V) =\P_{\ps{S}_V}\nonumber
\end{eqnarray}
for all $\ps{S}\in\Subcl\Sig$, $V\in\VH$.

\item There is a bijection 
\begin{eqnarray}
                c:\Sub\G&\map&\HG\label{Def:c}\\
c(\ps{A})_V&:=&\bigvee\{\hat\a\mid\hat\a\in\ps{A}_V\}\nonumber
\end{eqnarray}
for $\ps{A}\in\Sub{\G}$, $V\in\VH$ (see Proposition \ref{Theorem:SubO=HypO}).

The inverse is
\begin{eqnarray}
d:\HG&\map&\Sub\G   \label{Def:d}\\
d(\ga)_V&:=&\{\hat\a\in\G_V\mid\hat\a\preceq\hat\ga_V\}\nonumber
\end{eqnarray}
for all $\ga\in\HG$, $V\in\VH$.
\end{enumerate}

\noindent It follows from the above that there  is a bijection
\begin{eqnarray}
f:\Sub\G&\map&\Subcl\Sig        \label{SubO=SubclSig}\\
   f(\ps{A})_V&:=&S_{\bigvee\{\hat\a\in \ps{A}_V\}}\nonumber
\end{eqnarray}
for all $\ps{A}\in\Sub\G$, $V\in\VH$. The inverse is
\begin{eqnarray}
        g:\Subcl\Sig&\map&\Sub\G\\
g(\ps{S})_V&:=&\{\hat\a\in\G_V\mid\hat\a\preceq\P_{\ps{S}_V}\}
\end{eqnarray}
for all $\ps{S}\in\Subcl\Sig$, $V\in\VH$. We simply define $f:=k\circ c$ and
$g:=d\circ j$ where the bijections $c:\Sub\G\map\HG$,
$k:\HG\map\Subcl\Sig$, $j:\Subcl{\Sig}\map\HG$, and
$d:\HG\map\Sub\G$ are defined in \eq{Def:c}, \eq{Def:k},
\eq{Def:j}, and \eq{Def:d} respectively.

The problem posed in \eq{GaX=Subcl} can now be solved:
\begin{proposition}
The global elements of the power-object presheaf $P\G$ correspond
bijectively with the clopen sub-objects of $\Sig$.
\end{proposition}
\begin{proof}
We have shown above that there is a bijective correspondence
\begin{equation}
        \Sub\G\simeq\Subcl\Sig\label{SubO=SubclSigb}
\end{equation}
Now for any object $\ps{A}$ a fundamental property of the power
object, $P\ps{A}$, is that $\Ga(P\ps{A})\simeq\Sub{\ps{A}}$. It
follows from \eq{SubO=SubclSig} that
\begin{equation}
                \Ga(P\G)\simeq\Subcl\Sig
\end{equation}
\end{proof}

At this point it is worth stating the well-known specific form of
a power object in a topos of presheaves. Specifically:
\begin{definition}\label{DefPA}
The power object $P\ps{A}$ of any presheaf $\ps{A}$ over $\VH$ is
the presheaf given by
\begin{itemize}
        \item [(i)] On objects $V\in\VH$, for all $V\in\VH$, $P\ps{A}_V$ is defined as
        \begin{equation}
           P\ps{A}_V:=\{n_V:\ps{A}\ddown{V}\map\Om{\ddown{V}}
           \mid n_V\text{ is a natural transformation}\}\label{DefPAV}
        \end{equation}
        where $\ps{A}_{\ddown{V}}$ is the restriction of the 
        presheaf $\ps{A}$ to the smaller poset $\down{V}\subset\VH$, and analogously for
$\Om{\ddown{ V}}$.
        \item [(b)] On morphisms $i_{V'V}:V'\map V$
        (\ie $V'\subseteq V$), the presheaf functions 
        $P\ps{A}(i_{V'V}):P\ps{A}_V \map P\ps{A}_{V'}$ are defined as
        \begin{align}
       P\ps{A}(i_{V'V}):P\ps{A}_V &\map P\ps{A}_{V'}\nonumber\\
        n_V &\mapsto n_V{\ddown V'}\label{Def:PAiV'V}
        \end{align}
        Here, ${n_V}{\ddown V'}$ is the obvious restriction of the natural
        transformation $n_V$ to a natural transformation
        $\ps{A}{\ddown{V'}}\map\Om{\ddown{V'}}$.
\end{itemize}
\end{definition}

We now prove the `internal' analogue of the `external' isomorphism
$ \Sub\G\simeq\Subcl\Sig$ in \eq{SubO=SubclSig}, namely:
\begin{theorem}\label{Theorem:PO=PclS}
In the presheaf topos $\SetC\VH$, there is an isomorphism
\begin{equation}
        P\G\simeq \PSig
\end{equation}
between the power objects $P\G$ and $\PSig$.
\end{theorem}
\begin{proof}

It follows from \eq{DefPAV} that an equivalent definition of
$P\ps{A}$ is
\begin{equation}
        P\ps{A}_V:=\Sub{\ps{A}{\ddown V}}\label{PAV=SubAV}
\end{equation}
for all $V\in\VH$. In this form, the presheaf maps in
\eq{Def:PAiV'V} are just the restriction of an object in
$\Sub{\ps{A}{\ddown V}}$ to $\down{V'}$ for all $V'\subseteq V$.

In particular, the Definition \ref{DefPA} (and equation \eq{PAV=SubAV})
applies to the power object, $P\G$, of the outer presheaf $\G$. The
presheaf $\PSig$ is defined in the same way except only clopen sub-objects of $\Sig$ are used. Thus, for all $V\in\VH$, we have
\begin{equation}
        P\G_V=\Sub{\G{\ddown V}}
\end{equation}
and
\begin{equation}
        \PSig_{V}=\Subcl{\Sig{\ddown V}}
\end{equation}
However, the bijection $f:\Sub\G\map\Subcl\Sig$ in
\eq{SubO=SubclSig} clearly gives rise to a series of `local'
bijections
\begin{equation}
f_V:\Sub{\G{\ddown V}}\map\Subcl{\Sig{\ddown V}}
\end{equation}
for all $V\in\VH$. The inverse is
\begin{equation}
g_V:\Subcl{\Sig{\ddown V}}\map\Sub{\G{\ddown V}}
\end{equation}
These local maps are consistent with subspace inclusions
$i_{V'V}:V'\subseteq V$, \ie they are the components of two (mutually inverse) natural transformations
\begin{equation}
                        f:P\G\map\PSig;\qquad g:\PSig\map P\G.
\end{equation}
Therefore, $P\G$ can be identified
with the clopen power object $\PSig$.
\end{proof}

We note that we we also have the bijections
\begin{equation}
c_V:\Sub{\G\ddown{V}}\map{\rm Hyp}(\G\ddown V)
\end{equation}
with inverse
\begin{equation}
d_V:{\rm Hyp}(\G\ddown V)\map\Sub{\G\ddown{V}}.
\end{equation}

\subsection{Generalised truth objects}
\label{Subsec_GeneralisedTruthObjs} We first recall from
\eq{TTVPpsi} that the topos truth value in $\Ga\Om$ for a proposition
represented by a projection operator $\P$ is
\begin{equation}
\TVal{\P}{\ket\psi}(V):=\{V'\subseteq V\mid\bra\psi
                \delta(\P)_{V'}\ket\psi=1\}
                \label{TTVPpsi2}
\end{equation}
for all $V\in\VH$. The expression in \eq{TTVPpsi2} can 
be usefully rewritten as
\begin{equation}
\TVal{\P}{\ket\psi}(V):=\{V'\subseteq V\mid
                \delta(\P)_{V'}\succeq\ketbra\psi\}
                \label{TTVPpsi3}
\end{equation}
and, similarly, the  `original' truth object  given in \eq{Def:TOpsi} can be rewritten as as
\begin{equation}
\TOO^{\ket\psi}_V=\{\hat\a\in\G_V\mid\hat\a\succeq\ketbra\psi\}
  \label{Def:TOorgpsi}
\end{equation}

As defined in \eq{Def:TOorgpsi}, ${}^{\rm org}\ps\TO^{\ket\psi}$  is
a sub-object of $\G$. The mathematical expression $\Val{\ps{\de(\P)}\in
{}^{\rm org}\ps\TO^{\ket\psi}}$ only has meaning if ${}^{\rm
org}\ps\TO^{\ket\psi}$ can be regarded a sub-object of $\PSig$, which it can by
virtue of the monic arrow $\G\map\PSig$. However, in order to give meaning to the valuations
$\Val{\ps{S}\in\ps\TO^{\ket\psi}}$ for an arbitrary clopen
sub-object, $\ps{S}$ of $\Sig$ we need to find a new expression
for a truth object that does not `factor' through the monic $\G\map\PSig$. 

Using the isomorphism $P\G\simeq \PSig$ proved in Proposition
\ref{Theorem:PO=PclS}, the new truth object can be regarded as a
sub-object of $P\G$. The concrete form of $P\G$ can be written in several equivalent ways using the local isomorphisms 
\begin{equation}
        P\G_V=\Sub{\G\ddown{V}}\simeq{\rm Hyp}(\G\ddown{V})\simeq
        \Subcl{\Sig\ddown{V}}.
\end{equation}
Specifically, we define the new truth object
$\ps\TO^{\ket\psi}\subset P\G$ as
\begin{eqnarray}
\ps\TO^{\ket\psi}_V&:=&\{\ps{A}\in\Sub{\G\ddown{V}}
\mid\forall V'\subseteq V,\ketbra\psi\in\ps{A}_{V'}\}\label{TVpsiSubO}\\
&\simeq&\{\ga\in{\rm Hyp}(\G\ddown{V})\mid \forall V'\subseteq V,
\ketbra\psi\preceq\hat\ga_{V'}\}\label{TVpsiHyp}\\
&\simeq&\{\ps{S}\in\Subcl{\Sig\ddown{V}}\mid\forall V'\subseteq V,
\ketbra\psi\preceq \P_{\ps{S}_{V'}}\}\label{TVpsiSubcl}
\end{eqnarray}
for all $V\in\VH$. For inclusions $i_{V'V}:V'\subseteq V$, the
presheaf maps $\ps\TO^{\ket\psi}(i_{V'V}):\ps\TO^{\ket\psi}_V \map
\ps\TO^{\ket\psi}_{V'}$ are just the obvious restrictions from
$\down{V}$ to $\down{V'}$.

We note that,\footnote{We also note that 
$\TOO^{\ket\psi}\in\Ga\ps\TO^{\ket\psi}$.} using \eq{TVpsiSubO},
\begin{equation}
        \Ga\ps\TO^{\ket\psi}\simeq\{\ps{A}\in\Sub{\G}\mid
        \forall V\in\VH, \ketbra\psi\in\ps{A}_V\}
\end{equation}
with equivalent expressions using \eq{TVpsiHyp} and \eq{TVpsiSubcl}.
In particular, 
\begin{equation}
\Ga\ps\TO^{\ket\psi}\simeq\{\ps{S}\in\Subcl\Sig\mid
\forall V\in\VH, \ketbra\psi \preceq\P_{\ps S_V}\}
\end{equation}

The topos truth value in
\eq{TTVpsi} can now be rewritten as
\begin{equation}
\TVal{\Ain\De}{\ket\psi}= \Val{\ps{\delta(\hat
E[\Ain\De])}\in\ps{\TO}^{\ket\psi}}
\end{equation}
However, we can now also give meaning to the valuation
$\Val{\ps{S}\in\ps\TO^{\ket\psi}}$ for \emph{any} clopen subset,
$\ps{S}$ of $\Sig$, not just those of the form $\ps{\de(\P)}$ for
some projection operator $\P$ on $\Hi$.

\subsection{The truth objects  $\ps\TO^{\rho,r}$, $r\in(0,1]$}
The truth objects will now be generalised by considering the one-parameter family
$\TVal{-}{\rho}^r$ of valuations given by \eq{TTVrho_r}. For each density matrix, $\rho$, and $r\in[0,1]$ we will
associate a corresponding truth object $\ps\TO^{\rho,r}$.

Note that, strictly speaking, if $r<1$ then $\ps\TO^{\rho,r}$ should not be
called a `truth' object: its global elements  represent
propositions that are only true \emph{with probability at least $r$}
in the state $\rho$. Although
$\ps\TO^{\ket\psi, 1}$ turns out to be the truth object defined in \eq{TVpsiSubcl}, in general, for
$0<r<1$, we will get collections of propositions that are not
totally true in the state $\ket\psi$.\footnote{We remark that mathematically, it is no problem to include the probability $r=0$, but from an interpretational viewpoint, one may want to exclude it, since global sections of the truth object $\ps\TO^{\rho,0}$ represent propositions that are true with probability at least $0$, and every proposition fulfils this trivially.}

The idea is to generalise condition \eq{TVpsiSubcl}, which
played the key role in the definition of $\ps\TO^{\ket\psi}$.
This condition determines which projections can appear as
components of global elements of the truth object
$\ps\TO^{\ket\psi}$. Thus this condition implements the
requirement that a sub-object $\ps S\subseteq\Sig$ which represents
a proposition that is totally true in the state $\ket\psi$ must
have components $\ps S_V$ that are `true from the local
perspective $V$', for all local perspectives (\ie contexts)
$V\in\VH$. Each such component $\ps S_V$ corresponds
to a projection $\P_{\ps S_V}$ in $V$ and represents a proposition
that is available in the context $V$. This local proposition is
(locally) true in the state $\ket\psi$ if and only if 
$\ketbra\psi\preceq\P_{\ps S_V}$ holds for the projections.

We can now generalise this condition. If $\rho_{\ket\psi}=\ketbra\psi$ is the
density matrix corresponding to the pure state $\ket\psi$ then
\begin{equation}
\ketbra\psi\preceq\P \ \Longleftrightarrow\ \
\tr(\rho_{\ket\psi}\P)=1.
\end{equation}
Instead of demanding $\tr(\rho_{\ket\psi}\P)=1$, we now just require
$\tr(\rho_{\ket\psi}\P)\geq r$, for a given $r\in [0,1]$. (For $r=0$,
the condition is trivially true.) We can also extend the idea to
general mixed states $\rho$ and demand that $\tr(\rho\P)\geq r$.

The generalised truth object $\ps\TO^{\rho,r}$ is defined using $\ps\TO^{\ket\psi}$ in \eq{TVpsiSubcl} as an analogue. Thus
\begin{equation}
\ps\TO^{\rho,r}_V:=\{\ps{S}\in\Subcl{\Sig\ddown{V}}\mid
\forall V'\subseteq V,\tr(\rho\P_{\ps S_{V'}})\geq r\}
\label{Def:TVNew}
\end{equation}
for all $V\in\VH$. For inclusions $i_{V'V}: V'\subseteq V$,  the presheaf maps  $\ps\TO^{\rho,r}(i_{V'V}):\ps\TO^{\rho,r}_V \map \ps\TO^{\rho,r}_{V'}$ are defined as the obvious restriction of sub-objects from $\down{V}$ to $\down{V'}$. Clearly, $\ps\TO^{\rho,r}$ is a sub-object of the power object $P\G$ . We note that $\ps\TO^{\ket\psi}=\ps\TO^{\rho_\psi,1}$.

\begin{lemma} For all states $\rho$ and all real coefficients
(\ie probabilities) $0\leq r_1<r_2\leq 1$, we have
\begin{equation}
\ps\TO^{\rho,r_1}\supseteq\ps\TO^{\rho,r_2}.
\end{equation}
\end{lemma}

\begin{proof}
The assertion is that for all $V\in\VH$
\begin{equation}
\ps\TO^{\rho,r_1}_V\supseteq\ps\TO^{\rho,r_2}_V.
\end{equation}
By the definition in  \eq{Def:TVNew}, $\ps\TO^{\rho,r_2}_V$ is the family 
of clopen sub-objects $\ps S$ of $\Sig_{\downarrow V}$ such that for all
$V'\in\down{V}$ we have $\tr(\rho\P_{\ps
S_{V'}})\geq r_2$. But since $r_1<r_2$, $\tr(\rho\P_{\ps
S_{V'}})\geq r_2$  implies $\tr(\rho\P_{\ps S_{V'}})> r_1$
for all $V'\in\downarrow\!\!V$. Hence,
$\ps\TO^{\rho,r_2}_V\subseteq \ps\TO^{\rho,r_1}_V$.
\end{proof}

This result expresses the fact that the collection of sub-objects
that represent propositions which are true with probability at
least $r_1$ is bigger than the collection of sub-objects which
represents propositions that are true with probability at least
$r_2>r_1$. However, the generalised truth objects
$\ps\TO^{\rho,r_1}, \ps\TO^{\rho,r_2}$ are only collections of
sub-objects locally (in $V$). Globally they are presheaves whose global elements are sub-objects
of $\Sig$ that represent propositions
which are true with probability at least $r_1$, resp. $r_2$, in the
state $\rho$.
Note that the result $\ps\TO^{\rho,r_2}\subseteq\ps\TO^{\rho,r_1}\subseteq P\G$ implies
that $\Ga\ps\TO^{\rho,r_2}\subseteq\Ga\ps\TO^{\rho,r_1}$.

It can now be shown that
\begin{equation}
\TVal{\Ain\De}{\rho}^r=\Val{\ps{\delta(\hat
E[\Ain\De])}\in\ps{\TO}^{\rho, r}} \label{nuAinDer}
\end{equation}
for all $r\in [0,1]$ where, as in \eq{TTVrho_r}, for all stages
$V\in\VH$, we define
\begin{equation}
\TVal{\Ain\De}{\rho}^r(V):=\{V'\subseteq
V\mid\tr\big(\rho\,\delta(\hat E[\Ain\De])_{V'}\big)\geq r\}.
                \label{TTVrho_r2}
\end{equation}

There is no (obvious) analogue for $\TVal{\Ain\De}{\rho}^r$ of
the pseudo-state option in \eq{nupseudo} for
$\TVal{\Ain\De}{\ket\psi}$ and hence, in what follows, we will focus on
the use of truth objects.

\subsection{A topos perspective on classical probability theory}
\label{ToposClassicalProbs}
Before we proceed to develop  a topos version of quantum
probabilities, we first sketch a topos perspective on classical
probability theory. This is interesting in its own right and also provides
guidance for the quantum case later on.

Thus, suppose $X$ is a space with a
probability measure $\mu:\Sub{X} \map[0,1]$. Here, $\Sub{X}$
denotes the $\mu$-measurable subsets of $X$. Then we wish to
describe this situation using the topos $\ShL{(0,1)}$ discussed in
Section \ref{Sec_ToposRepOfProbabs}. As a first step we define the
sheaf $\ps{X}$ as the \'etale bundle over $(0,1)_L$ with constant
stalk $X$ over each $(0,r)\in {\cal O}((0,1)_L)$; in terms of the usual notation, $\ps{X}:=\Delta X$. Similarly, for any measurable subset $S\subseteq X$
of $X$ we define $\ps{S}:=\Delta S$. Thus there is a map
\begin{eqnarray}
        \De:\Sub{X}&\map& {\rm Sub}_{\ShL{(0,1)}}(\ps X)\nonumber\\
           S&\mapsto& \ps{S}:=\De S.\label{Def:De}
\end{eqnarray}

Motivated by \eq{Def:TVNew} we then define the sheaf
$\ps{\TO}^\mu$ in $\ShL{(0,1)}$ by
\begin{equation}
   \ps\TO^\mu_{(0,r)}:=\{S\subseteq X\mid\mu(S)\geq r\}
\end{equation}
at all stages $(0,r)\in {\cal O}((0,1)_L)$. Using the isomorphism \eq{beta}, we denote the stages as $r$, where $r\in [0,1]$:
\begin{equation}
                        \ps\TO^\mu_r:=\{S\subseteq X\mid\mu(S)\geq r\}.
\end{equation}
Note that, here, $r$ labels the stages, while in \eq{Def:TVNew} $r$ is fixed. If
$r_1<r_2$, the presheaf maps $\ps\TO^\mu(r_1<r_2):\ps\TO^\mu_{r_2} \map \ps\TO^\mu_{r_1}$ are defined in the obvious (trivial)
way:
\begin{align}
\ps\TO^\mu(r_1<r_2):\ps\TO^\mu_{r_2} &\map \ps\TO^\mu_{r_1}\\
                        S &\mapsto S.
\end{align}
It is easy to show that this defines a sheaf over the topological space $(0,1)_L$.

Now, $\ps{S}$ is a sub-object of $\ps{X}$, and $\ps\TO^\mu$ is  sub-object of $P\ps{X}$.\footnote{In the context of a formal language for the system, ${\TO}^\mu$ is of type $PP{X}$ and any $\De S$ is of type $P{X}$ (see \cite{DI(1), DI(Coecke)08} for more details on the use of types).} Therefore, the valuation
$\Val{\ps{S}\in\ps\TO^\mu}$ is well-defined as an element of
$\Ga\Om^{(0,1)}$. Specifically, at each stage $r\in [0,1]$ we have
\begin{eqnarray}
\Val{\ps S\in \ps{\TO}^\mu}(r)&:=&\{r'\leq r\mid \ps S_{r'}\in
                        \ps{\TO}^\mu_{r'}\}\nonumber\\
        &=&\{r'\leq r\mid \mu(S)\geq r'\}\label{Def:SinTmua}\\
        &=& [0,\mu(S)]\cap (0,r]        \label{Def:SinTmub}\\
        &=& [0,\min\{\mu(S),r\}],
\end{eqnarray}
which, as required, belongs to $\Om^{(0,1)}_r=\Om^{(0,1)}_{(0,r)}$. The resemblance to
\eq{TTVrho_r2} is clear and suggests that in the
quantum-theoretical expression the parameter $r$
should be viewed as  a `stage of truth' associated with the topos $\ShL{(0,1)}$. We  return to this in the next Section.

It is clear from \eq{Def:SinTmua} that, for any subset
$S\subseteq X$, the value of $\mu(S)$ can be recovered from the
valuation/global element $\Val{\ps S\in
\ps{\TO}^\mu}\in\Ga\Om^{(0,1]}$ in the topos $\ShL{(0,1]}$. More
precisely, the map $\mu:\Sub{X}\map[0,1]$ determines, and is
determined by, the map $\xi^\mu:{\rm Sub}_{{\ShL{(0,1)}}}(\ps
X)\map\Ga\Om^{(0,1]}$ defined as\begin{equation}
                        \xi^\mu(\ps S):=\Val{\ps S\in \ps{\TO}^\mu}\label{Def:ximu}
\end{equation}
for all $\ps S\in{\rm Sub}_{{\ShL{(0,1)}}}(\ps X)$. Thus we have
replaced the measure  $\mu$ with the collection of truth values
$\Val{\ps S\in \ps{\TO}^\mu}\in\Ga\Om^{(0,1)}$, where $S\subseteq
X$. Of course, in the present case this is rather trivial, but it
gives insight into how to proceed in the quantum theory.

Actually, the classical  result has some interest in its own right.
Essentially we have a new
interpretation of classical probability in which the results are
expressed in terms of truth values in the sheaf topos $\ShL{(0,1)}$. 
This suggests a `realist' (or `neo-realist')
interpretation of probability theory that could replace the
conventional (for a scientist) instrumentalist interpretation in
terms of relative frequencies of measurements.

We also remark that the truth values, given by $\Ga\Om^{(0,1)}$, form a Heyting algebra, which is one aspect of the intuitionistic logic of the topos. 
One well-known realist view of probability is the
\emph{propensity} theory in which a probability is viewed as the
propensity, or tendency, or potentiality, of the associated event
to occur. Thus the results above might be used to give a
precise mathematical definition of propensities in terms of the
Heyting algebra $\Ga\Om^{(0,1)}$, though it should be remarked that the role of the Heyting algebra structure is not entirely clear at the moment.

Using the map $\ell:[0,1]\map\Ga\Om^{(0,1)}$ defined in \eq{Def:ka(a)(r)}, we obtain
\begin{eqnarray}
                        [\ell\circ\mu(S)](r) &=& \ell(\mu(S))(r)=\{(0,r')\subseteq (0,r) \mid \mu(S)\geq r'\}\nonumber\\
                        &=&\Val{\ps S\in \ps{\TO}^\mu}(r)=\xi^\mu(\De(S))(r)
\end{eqnarray}
for all $r\in[0,1]$. Thus
\begin{equation}
        \ell\circ\mu=\xi^\mu\circ\Delta
\end{equation}
in the commutative diagram
\begin{equation}
\begin{diagram}                 \label{Def:CommDiagramClassical}
        \Sub{X} & & \rTo^{\mu} & & [0,1] \\
        \dTo^{\Delta} & & & & \dTo_{\ell} \\
        {\rm Sub}_{\ShL{(0,1)}}(\ps X) & & \rTo_{\xi^\mu} & & \Ga\Om^{(0,1)}
\end{diagram}
\end{equation}
This diagram will provide an important analogue in the following discussion of the quantum theory.

A measure $\mu$ is usually taken to be $\sigma$-additive, that is, for any countable family $(S_i)_{i\in\bbN}$ of pairwise disjoint, measurable subsets of $X$, we have
\begin{equation}
                        \mu(\bigcup_i S_i)=\sum_i \mu(S_i).
\end{equation}
We reformulate this property slightly. Let $\tilde S_0:=S_0$, and for each $i>0$, define recursively
\begin{equation}
                        \tilde S_i:=\tilde S_{i-1}\cup S_i.
\end{equation}
Then $(\tilde S_i)_{i\in\bbN}$ is a countable increasing family of measurable subsets of $X$. Now $\sigma$-additivity can be expressed as
\begin{equation}                        \label{Eq_sigmaAddAsJoinPreservation}
                        \mu(\bigvee_i \tilde S_i)=\mu(\bigcup_i \tilde S_i)=\sup_i \mu(\tilde S_i)=\bigvee_i \mu(\tilde S_i),
\end{equation}
that is, $\mu$ preserves countable joins (suprema).

The map $\ell:[0,1]\map\Ga\Om^{(0,1)}$, defined in \eq{ell}, also preserves joins, as can be seen easily: let $(p_i)_{i\in I}$ be a family of real numbers in the unit interval $[0,1]$ (the family need not be countable). Then
\begin{equation}
        \ell(\bigvee_i p_i)_r=\ell(\sup_i p_i)_r=\left\{\begin{array}{ll}
                                                                                                                        [0,r]=\Om_r^{(0,1)} &\mbox{ if $\sup_i p_i\geq r$}\\
                                                                $[$0,\sup_i p_i] &\mbox{ if $0<\sup_i p_i<r$}\\
                                                          \emptyset&\mbox{ if $\sup_i p_i=0$}
                                                           \end{array}
                                                   \right.
\end{equation}
and
\begin{equation}
        \bigvee_i \ell(p_i)_r=\bigvee_i\left\{\begin{array}{ll}
                                                                                                                        [0,r]=\Om_r^{(0,1)} &\mbox{ if $p_i\geq r$}\\
                                                                $[$0,p_i] &\mbox{ if $0<p_i<r$}\\
                                                          \emptyset&\mbox{ if $p_i=0,$}
                                                           \end{array}
                                                   \right.
\end{equation}
which directly implies $\ell(\bigvee_i p_i)_r=\bigvee_i(\ell(p_i)_r)$ for all $r\in [0,1]$. Since joins are defined stage-wise in $\Ga\Om^{(0,1)}$, we obtain
\begin{equation}
                        \ell(\bigvee_i p_i)=\bigvee_i \ell(p_i).
\end{equation}
Hence, the composite map $\ell\circ\mu:\Sub{X}\map\Ga\Om^{(0,1)}$ preserves countable joins of increasing families $(\tilde S_i)_{i\in\bbN}$ of measurable subsets, that is,
\begin{equation}
                        (\ell\circ\mu)(\bigvee_i \tilde S_i)=\bigvee_i (\ell\circ\mu)(\tilde S_i).
\end{equation}
Note that $\ell\circ\mu$ corresponds to the upper-right path through the diagram \eq{Def:CommDiagramClassical}. Since the diagram commutes, the left-lower path, \ie the map $\xi^\mu\circ\Delta$, also preserves such joins,
\begin{equation}
                        (\xi^\mu\circ\Delta)(\bigvee_i \tilde S_i)=(\ell\circ\mu)(\bigvee_i \tilde S_i)=\bigvee_i (\ell\circ\mu)(\tilde S_i).
\end{equation}
This is the logical reformulation of $\sigma$-additivity of the measure $\mu$.

\section{Application to Quantum Theory}
\subsection{The maps $j_r:\Ga\UI\rightarrow\Ga\Om$, $r\in(0,1]$}
These ideas will now be applied to quantum theory, keeping in mind
the commutative diagram in \eq{Def:CommDiagramClassical} whose
quantum analogue we seek. We already have the analogue of
$\mu:\Sub{X}\map[0,1]$ in the form of the measures
$\mu^\rho:\Subcl\Sig\map\Ga\UI$ where $\rho$ is a density matrix.
The challenge is to fill in the rest of the diagram for the
quantum case.

The first step is to see if $\Ga\UI$ can be associated with
$\Ga\Om$ in some way. To do this note that the
`$0$' and `$1$' that appear in $\Ga\UI$ refer to probability $0$
and $1$, and in that sense correspond to `false' and `true'
respectively. To clarify this we first define, for each
$r\in[0,1]$, the global element $\ga_r\in\Ga\UI$ as
\begin{equation}
        \ga_r(V):=r
\end{equation}
for all $V\in\V{}$. Then $\mu^\rho(\emptyset)=\ga_0$ and
$\mu^\rho(\Sig)=\ga_1$, which suggests  the sections $\ga_0$ and
$\ga_1$  are to be associated respectively with $\bot_\Om$ and $\top_\Om$
in $\Ga\Om$ .

Next define a map $j:\Ga\UI\rightarrow \Ga\Om$ by
\begin{equation}
        j(\ga)(V):=\{V'\subseteq V\mid \ga(V')=1\}\label{Def:j(ga)}
\end{equation}
for $V\in\VH$. This is a  sieve on $V$ since once $\ga(V')$
becomes `true' (\ie equal to $1$)\ it remains so because $\ga$ is a
nowhere-increasing function. Note that if $\ga_r$ denotes the
`constant' section with value $r$ then
\begin{equation}
        j(\ga_r)(V):=\{V'\subseteq V\mid r=1\}
\end{equation}
so that
\begin{equation}
 j(\ga_r)=\left\{\begin{array}{ll}
                  \top_{\Om} &\mbox{ if $r=1$ }\\
                   \bot_{\Om} &\mbox{ if $r<1$}
                         \end{array}
                    \right.
\end{equation}
Note also that, to each density matrix $\rho$ and projection
operator $\P$, there corresponds a global  element
$\ga_{\P,\rho}\in\Ga\UI$ defined by
\begin{equation}
 \ga_{\P,\rho}(V):={\rm tr}\big(\rho\,\delta(\P)_V\big)
        \label{Def:gaPrho}
\end{equation}
for all $V\in\VH$. Hence
\begin{equation}
j(\ga_{\P,\rho})(V)=\{V'\subseteq V\mid {\rm
tr}\big(\rho\,\delta(\P)_V\big)=1\}\label{j(gaPrho(V)}.
\end{equation}Our intention is to consider $\Ga\UI$ as the quantum analogue of
$[0,1]$ in the diagram \eq{Def:CommDiagramClassical}. 

Returning to the definition \eq{Def:j(ga)} of $j:\Ga\UI\map
\Ga\Om$,  we now make the critical observation that there is a
natural one-parameter family of maps
$j_r:\Ga\UI\map\Ga\Om$, $r\in(0,1]$, defined by
\begin{equation}
        j_r(\ga)(V):=\{V'\subseteq V\mid \ga(V')\geq r\}
        \label{Def:jr}
\end{equation}
for all stages $V\in\VH$. Thus \eq{Def:j(ga)} is just the special
case $j=j_1$.

It is important to note that the parameter $r$ in \eq{Def:jr} has
been chosen to lie in $(0,1]$ rather than $[0,1]$. This is
because, for all $\ga\in\Ga\UI$, for $r=0$ we have $\{V'\subseteq
V\mid \ga(V')\geq 0\}=\down{V}=\top_{\Om_V}$ for all $V\in\VH$. In particular,
even if $\ga=0$ we would still assign the truth value `totally true'.

To interpret this, suppose we choose $\De$ to lie completely \emph{outside} the spectrum of $\hat A$. Then $\hat E[\Ain\De]=\hat 0$, the null projection representing the trivially false proposition in standard quantum theory. $\hat E[\Ain\De]=\hat 0$ implies that $\delta(\hat E[\Ain\De])_V=\hat 0$ for all $V$, which corresponds to the empty subobject $\ps{\emptyset}\in\Subcl{\Sig}$, the representative of the trivially false proposition in the topos approach to quantum theory. For any quantum state $\rho$, the corresponding probability measure $\mu^\rho:\Subcl{\Sig}\map\Ga\UI$ will map the empty subobject $\ps\emptyset$ to the global element $\ga_{\hat 0,\rho}=0$, the global element that is constantly $0$ (see \eq{Def:murho} and \eq{Def:gaPrho}; for details, see \cite{ADmeas}). From \eq{Def:jr}, we obtain
\begin{equation}
                        j_r(\ga_{\hat 0,\rho})(V)=\emptyset
\end{equation}
for all $r>0$, while for $r=0$,
\begin{equation}
                        j_0(\ga_{\hat 0,\rho})(V)=\top_{\Om_V}.
\end{equation}
The latter means that we assign `totally true' to the trivially false proposition, which we want to avoid by excluding $r=0$. Note that \emph{only} for $r=0$, we get `totally true', for all $r>0$, we obtain `totally false'. It is a matter of interpretation if one wants to allow that even the trivially false proposition can be true with probability $0$ (since \emph{any} proposition is true with at least probability $0$), or if the proposition that conventionally is interpreted as trivially false must be totally false. In the latter case, one must exclude the case $r=0$, as we will do here. This has no bearing on our results.

For later reference, we note that in quantum theory, 
$\ga_{\P,\rho}(V):= \tr(\rho\,\delta(\P)_{V})$ (see \eq{Def:gaPrho}) and hence
\begin{equation}
 j_r(\ga_{\P,\rho})(V)=
 \{V'\subseteq V\mid \tr(\rho\,\delta(\P)_{V'})\geq r\}
\end{equation}
so that, in particular,
\begin{eqnarray}
                        j_r(\ga_{\hat E[\Ain\De],\rho})(V) &=&                  \nonumber
                        \{V'\subseteq V\mid \tr(\rho\,\delta(\hat E[\Ain\De])_{V'})\geq r\}\\
                        &=& \TVal{\Ain\De}{\rho}^r(V)  \label{jrga=nuAinDe}
\end{eqnarray}
for all $V\in\VH$. 

Now we come to an important result.
\begin{theorem}\label{jrseparates}
The family of maps $j_r:\Ga\UI\map\Ga\Om$, $r\in(0,1]$, is
\emph{separating}: \ie if, for all $r\in(0,1]$ we have
$j_r(\ga_1)=j_r(\ga_2)$ then $\ga_1=\ga_2$.
\end{theorem}

\begin{proof}
Suppose there are $\ga_1$ and $\ga_2$ in $\Ga\UI$ such that $\ga_1\neq\ga_2$.
Then there exists $V_0\in\VH$ such that
$\ga_1(V_0)\neq\ga_2(V_0)$. Without loss of generality we can
assume that $\ga_1(V_0)>\ga_2(V_0)$. Now consider
\begin{equation}
j_{\ga_1(V_0)}(\ga_2)(V_0)=\{V'\subseteq
V_0\mid\ga_2(V')\geq\ga_1(V_0)\}
        \label{jga1ga2V0}
\end{equation}
and
\begin{equation}
j_{\ga_1(V_0)}(\ga_1)(V_0)=\{V'\subseteq
V_0\mid\ga_1(V')\geq\ga_1(V_0)\}
        \label{jga1ga1V0}
\end{equation}
Equation \eq{jga1ga1V0} shows that $j_{\ga_1(V_0)}(\ga_1)(V_0)$ is
 the principal sieve $\down V_0$ on $V_0$. On the other
hand, since $\ga_1(V_0)>\ga_2(V_0)$, $\ga_2(V_0)$ cannot belong to
the sieve $j_{\ga_1(V_0)}(\ga_1)(V_0)$.  It follows that
\begin{equation}
        j_{\ga_1(V_0)}(\ga_1)\neq j_{\ga_1(V_0)}(\ga_2)
\end{equation}
and hence the family of maps $\{j_r:\Ga\UI\map\Ga\Om\mid
r\in(0,1]\}$ is separating.
\end{proof}


\subsection{The map $\ell:\Ga\UI\map\Ga\Om^{\Sh{\VH_A\times(0,1)_L}}$}                  \label{SubSec_Qell}
As shown above, the one-parameter family of maps
$j_r:\Ga\UI\map\Ga\Om$, $r\in(0,1]$, defined as
\begin{equation}
        j_r(\ga)(V):=\{V'\subseteq V\mid \ga(V')\geq r\}
                            \label{Def:jr2}
\end{equation}
separates the members of $\Ga\UI$ and, therefore, density matrices.
Our goal now is to combine this parameterised family to form a
\emph{single} map from $\Ga\UI$ to $\Ga\Om^\tau$ where $\Om^\tau$ is
the sub-object classifier of some new topos $\tau$.

The first step in this direction is the result gained from combining \eq{beta} and \eq{hsimeqGaO}, giving the isomorphism
\begin{equation}
        \ell:[0,1]\map\Ga\Om^{(0,1)}\label{hsimeqGaO_2}
\end{equation}
which, as discussed in Section \ref{ToposClassicalProbs},
enables a topos-logic interpretation to be given to classical
probability theory. The key results are summarised in the
commutative diagram in \eq{Def:CommDiagramClassical}.

We remark that, on the face of it, \eq{hsimeqGaO_2} is an isomorphism of sets, not Heyting algebras. By the isomorphism \eq{beta}, we could regard $[0,1]$ as a Heyting algebra, since ${\cal O}((0,1_L))$ is a Heyting algebra (as the open sets of any topological space form a Heyting algebra). Since \eq{hsimeqGaO} is a Heyting algebra isomorphism, too, we could see \eq{hsimeqGaO_2} as an isomorphism of (complete) Heyting algebras.

Yet, there is little to be gained from this: the measure $\mu:\Sub{X}\map [0,1]$ is not a Heyting algebra morphism, so in the commutative diagram \eq{Def:CommDiagramClassical}, we have some morphisms which are Heyting algebra morphisms, while others are not. We consider the fact that $[0,1]$ can be seen as a Heyting algebra, and $\ell$ as a Heyting algebra isomorphism, as coincidental. It is more important that $\ell$ is a set isomorphism, and is obviously order-preserving. This means that $\ell$ represents probabilities faithfully as truth values in our sheaf topos. Moreover, there are no other truth values apart from those corresponding with probabilities.

Yet, there is aspect of $\ell$ being a morphism of complete Heyting algebras that we \emph{do} use: namely, the preservation of joins which plays a key role in the logical reformulation of the $\sigma$-additivity of a measure $\mu$, as described at the end of subsection \ref{ToposClassicalProbs}.

In the quantum case, we can give the set of `generalised probabilities', \ie the codomain $\Ga\UI$ of our probability measures, the structure of a partially ordered set in the obvious way: let $\ga_1,\ga_2\in\Ga\UI$, then
\begin{equation}
            \ga_1\preceq\ga_2:\Longleftrightarrow\ga_1(V)\leq\ga_2(V)
\end{equation}
for all $V\in\VH$.

Let us now consider again the one-parameter family of topos truth
values
\begin{equation}
\TVal{\Ain\De}{\rho}^r(V):=\{V'\subseteq V\mid
\tr(\rho\,\delta(\hat E[\Ain\De])_{V'})\geq r\}\label{TTVrho_r3}
\end{equation}
for all stages $V\in\VH$. Here, $r\in(0,1]$ and, as shown in
\eq{jrga=nuAinDe}, \begin{equation}
        \TVal{\Ain\De}{\rho}^r=j_r(\ga_{\hat E[\Ain\De],\rho})
\end{equation}
We recall also the result from classical probability theory in
\eq{Def:SinTmua}
\begin{equation}
\Val{\ps{S}\in\TO^\mu}(r)=\{r'\leq r\mid\mu(S)\geq r'\}
\label{Def:SinTmu_c}
\end{equation}

Now, \eq{Def:SinTmu_c} is a sieve on the stage $r$ in the
sub-object classifier, $\Om^{\ShL{(0,1)}}$, in the topos
$\ShL{(0,1)}$, and \eq{TTVrho_r3} is a sieve on the stage $V$ in the
topos $\Sh{\VH_A}$. These results strongly suggest that the way of
`combining' the one-parameter family $j_r:\Ga\UI\map\Ga\Om$,
$r\in(0,1]$ is to use sieves on \emph{pairs}
$\lr{V}{r}\in\VH\times(0,1)$ with the ordering
\begin{equation}
\lr{V'}{r'}\preceq\lr{V}{r}\mbox{ iff } V'\subseteq V\mbox{ and }
r'\leq r.
\end{equation}

We recall from equation \eq{Eq_PresheafSheafTopos} that (i) $\VH_A$ denotes the context category, which is a poset, equipped with the lower Alexandroff topology, and (ii) there is an isomorphism of topoi $\SetC\VH\map\Sh{\VH_A}$.

We put the product topology on the poset $\VH_A\times (0,1)_L$. The basic opens sets in this topology are of the form
\begin{equation}
                        \down V\times (0,r)
\end{equation}
for $V\in\VH$ and $r\in [0,1]$. In the following, we will identify such a basic open with the pair $\la V,r\ra$, which makes some formulas easier to read.

Obviously, we are now interested in the sheaf topos $\Sh{\VH_A\times(0,1)_L}$, where we can 
define $\ell:\Ga\UI\map\Ga\Om^{\Sh{\VH_A\times(0,1)_L}}$ as
\begin{equation}
\ell(\ga)(\la V,r\ra):={\{}\la V',r'\ra\preceq \la
          V,r\ra\mid\ga(V')\geq r'\}  \label{Def:ell}
\end{equation}
That this \emph{is} a sieve on $\lr{V}{r}\in{\cal O}(\VH_A\times(0,1)_L)$
follows because $\ga$ is a nowhere-increasing function. In
particular, we have
\begin{eqnarray}
\hspace{-25pt}\ell(\ga_{\delta(\hat E[\Ain\De]),\rho)}(\la V,r\ra)&:= &\{\la
V',r'\ra\preceq \la V,r\ra\mid\ga_{\delta(\hat
E[\Ain\De]),\rho)}(V')
\geq r'\}\nonumber\\
        &=&\{\la V',r'\ra\preceq \la V,r\ra\mid
        \tr(\rho\,\delta(\hat E[\Ain\De])_{V'}\big)\geq r'\}
        \label{kgadEA}\ \
\end{eqnarray}

In this context it is important to note that the map
$\ell:\Ga\UI\map\Ga\Om^{\Sh{\VH_A\times(0,1)_L}}$ clearly
\emph{separates} the elements of $\Ga\UI$, as follows from an
obvious analogue of the proof of Theorem \ref{jrseparates}. For
suppose there are $\ga_1$ and $\ga_2$ such that $\ga_1\neq\ga_2$.
Then there exists $V_0\in\VH$ such that
$\ga_1(V_0)\neq\ga_2(V_0)$. Without loss of generality we can
assume that $\ga_1(V_0)>\ga_2(V_0)$. Then we compare
\begin{equation}
\ell(\ga_2)\big(\la V_0,\ga_2(V_0)\ra\big):=\{\la V',r'\ra\preceq
\la V_0,
        \ga_2(V_0)\ra \mid\ga_2(V')\geq r'\}=
        \; \downarrow\!\la V_0,\ga_2(V_0)\ra
\end{equation}
with
\begin{equation}
\ell(\ga_1)\big(\la V_0,\ga_2(V_0)\ra\big):=\{\la V',r'\ra\preceq
\la V_0,
        \ga_2(V_0)\ra \mid\ga_1(V')\geq r'\}\subsetneq
        \;\downarrow\!\la V_0,\ga_2(V_0)\ra
\end{equation}
It follows that $\ell(\ga_2)\neq\ell(\ga_1)$, as claimed. This means that the map $\ell:\Ga\UI\map\Ga\Om^{\Sh{\VH_A\times(0,1)_L}}$ is injective. 

We show that $\ell$ preserves joins: let $(\ga_i)_{i\in I}\subset\Ga\UI$ be a family of global elements of $\UI$, then
\begin{align}
                        \ell(\bigvee_i \ga_i)(\la V,r\ra) &=
                        \{\la V',r'\ra\preceq\la V,r\ra \mid (\bigvee_i \ga_i)(V')\geq r'\}\\
                        &= \{\la V',r'\ra\preceq\la V,r\ra \mid \sup_i\ga_i(V')\geq r'\}\\
                        &= \bigcup_i (\{\la V',r'\ra\preceq\la V,r\ra \mid \ga_i(V')\geq r'\})\\
                        &= \bigvee_i \ell(\ga_i)(\la V,r\ra).
\end{align}
Since this holds for all stages $\la V,r\ra$\footnote{Strictly speaking, the stages $\la V,r\ra=\down V\times (0,r)$ give just a basis of the topology on $\VH_A\times (0,1)_L$, so we would have to consider arbitrary unions of these sets to describe \emph{all} stages. This clearly poses no difficulty.} and joins in $\Ga\Om^{\Sh{\VH_A\times (0,1)_L}}$ are defined stagewise, we obtain
\begin{equation}
                        \ell(\bigvee_i \ga_i)=\bigvee_i \ell(\ga_i).
\end{equation}

\subsection{The truth object $\eps{\TO}^\rho$}
Our intention is that \eq{kgadEA} will be the truth value of the
proposition ``$\Ain\De$'' when the quantum state is the density
matrix $\rho$. The topos in question is $\Sh{\VH_A\times(0,1)_L}$
and, for notational convenience, we will denote a sheaf
on $\VH_A\times (0,1)_L$ by a symbol $\eps{A}$ to distinguish it from the symbol
$\ps{B}$ for a sheaf on $\VH_A$ (or, equivalently, a presheaf over $\VH$). A key step will be to
define a truth object $\eps{\TO}^\rho$ from which \eq{kgadEA}
follows as the correct truth value.

The first step is to construct certain physically important
sheaves over $\VH_A\times (0,1)_L$. To this end define the
projection maps $p_1:\VH_A\times (0,1)_L\map\VH_A$ and $p_2:\VH_A\times (0,1)_L\map (0,1)_L$. These
can be used to pull-back sheaves over $\VH_A$ or $(0,1)_L$ to give
sheaves over $\VH_A\times (0,1)_L$. Thus, if $\ps{C}$ is a sheaf on $\VH_A$
we have
\begin{equation}
        (p_1^*\ps{C})_\lr{V}{r}=\ps{C}_V
\end{equation}
at each stage $\lr{V}{r}\in{\cal O}(\VH_A\times (0,1)_L)$.

A key step is to identify the state object and quantity-value
object in this new topos $\Sh{\VH_A\times(0,1)_L}$. We will make
the simplest assumption that there is no $r$-dependence in the
state object: indeed, from a physical perspective it is hard to
see where such a dependence could come from. Therefore, we define
the state object in $\Sh{\VH_A\times(0,1)_L}$ as
\begin{equation}
        \eps\Si:=p_1^*\Sig.                     \label{Def:Sig__}
\end{equation}

In the topos $\Sh{\VH_A}$ a quantum proposition associated with the
projector $\P$ is represented by the
clopen sub-object, $\ps{\das{P}}$,
of $\Sig$. In the light of \eq{Def:Sig__} it is then natural to
assume that the sub-object that represents the quantum proposition
is just
\begin{equation}
 \eps{\das{P}}:=p_1^*\ps{\das{P}}.
\end{equation}
The quantity-value object must then be defined as
\begin{equation}
        \eps\R:=p_1^*\ps\R
\end{equation}
where $\ps\R$ is the quantity-value object in $\Sh{\VH_A}$. This
guarantees that the $\Sh{\VH_A}$-representation of a physical quantity
by an arrow $\breve{A}:\Sig\map\ps\R$ will translate into a
representation in $\Sh{\VH_A\times (0,1)_L}$ by an arrow from $\eps\Si$ to
$\eps\R$.

All the objects in $\Sh{\VH_A\times(0,1)_L}$ defined above are
obtained by a trivial pull-back from the corresponding objects in
$\Sh{\VH_A}$. The critical, non-trivial, object is the truth object
$\eps\TO^\rho$. One anticipates that this will be derived in some
way from the one-parameter family, $r\mapsto\ps\TO^{\rho,r}$,
$r\in(0,1]$ defined in \eq{Def:TVNew}.

In fact, the obvious choice works. Namely, define $\eps\TO^\rho$
as
\begin{equation}
       \eps\TO^\rho_\lr{V}{r}:=\ps\TO^{\rho,r}_V=
       \{\ps{S}\in\Subcl{\Sig\ddown{V}}\mid
\forall V'\subseteq V,\tr(\rho\P_{\ps S_{V'}})\geq r\}
\end{equation}
for all stages $\lr{V}{r}\in{\cal O}(\VH_A\times (0,1)_L)$.

This equation shows that $\eps{\TO}^\rho$ is a sub-object of
$P\eps{\Si}$, and $\eps{\das{P}}$ is a sub-object of $\eps{\Si}$.
It follows that the valuation $\Val{\eps{\das{P}}\in\eps\TO^\rho}$
is well-defined as an element of $\Ga\Om^{\Sh{\VH_A\times (0,1)_L}}$. The
value at any basic open $\lr{V}r=\down V\times (0,1)\in{\cal O}(\VH_A\times (0,1)_L)$ is
\begin{equation}
\Val{\eps{\das{P}}\in\eps\TO^\rho}\lr{V}{r}= \{\la V',r'\ra\preceq
\la V,r\ra\mid
        \tr(\rho\,\delta(\P)_{V'}\big)\geq r'\}\label{delPinTVr}
\end{equation}
which, as anticipated, is just \eq{kgadEA}. In summary, the truth
value associated with the proposition``$\Ain\De$'' in the quantum
state $\rho$\ is
\begin{equation}
\TVal{\Ain\De}{\rho}:=\Val{\eps{\de(\hat
E[\Ain\De])}\in\eps\TO^\rho}\in\Ga\Om^{\Sh{\VH_A\times (0,1)_L}}
        \label{TVinVHx}
\end{equation}

The equation \eq{TVinVHx} is the main result of our paper. The
critical feature of the truth value, \eq{TVinVHx}, in the topos
$\Sh{\VH_A\times(0,1)_L}$ is that it \emph{separates} density
matrices, in marked contrast to the more elementary truth value,
\eq{TTVrho}, in the topos $\Sh\VH$.

\subsection{The analogy with classical probability theory}
The analogy with classical probability is quite striking. This, we
recall, is summarised by the commutative diagram
\begin{equation}
\begin{diagram}                 \label{Def:CommDiagramClassical2}
        \Sub{X} & & \rTo^{\mu} & & [0,1] \\
        \dTo^{\Delta} & & & & \dTo_{\ell} \\
        {\rm Sub}_{\ShL{(0,1)}}(\ps X) & & \rTo_{\xi^\mu} & & \Ga\Om^{(0,1)}
\end{diagram}
\end{equation}

Our claim is that there is a precise quantum analogue of this,
namely the commutative diagram
\begin{equation}
\begin{diagram}                 \label{Def:CommDiagramQuantum}
        \Subcl{\Sig} & & \rTo^{\mu^\rho} & & \Ga\UI \\
        \dTo^{p_1^*} & & & & \dTo_{\ell} \\
        \Subcl{\eps\Si} & & \rTo_{\xi^\rho} & & \Ga\Om^{\Sh{\VH_A\times (0,1)_L}}
\end{diagram}
\end{equation}

The only map in \eq{Def:CommDiagramQuantum} that has not been
defined already is
\begin{equation}
                        \xi^\rho:\Subcl{\eps\Si}\map\Ga\Om^{\Sh{\VH_A\times (0,1)_L}}.
\end{equation}
Motivated by the construction \eq{Def:ximu} in the classical case,
we define
\begin{equation}
        \xi^\rho(\eps{S}):=\Val{\eps{S}\in\eps\TO^\rho}
\end{equation}
for all sub-objects $\eps{S}$ of $\eps\Si$. We now observe that
\eq{delPinTVr} can be rewritten as
\begin{equation}
\Val{p_1^*\ps{\das{P}}\in\eps\TO^\rho}\lr{V}{r}:=
\{\lr{V'}{r'}\preceq\lr{V}r\mid\mu^\rho(\ps{\das{P}}_{V'})\geq
r'\}
\end{equation}
where we have used that fact that
$\eps{\das{P}}:=p_1^*\ps{\das{P}}$.

However, $\ell:\Ga\UI\map\Ga\Om^{\Sh{\VH_A\times(0,1)_L}}$ was
defined in \eq{Def:ell} as
\begin{equation}
\ell(\ga)(\la V,r\ra):={\{}\la V',r'\ra\preceq \la
          V,r\ra\mid\ga(V')\geq r'\}  
\end{equation}
and so it follows that
\begin{equation}\xi^\rho(p_1^*\ps{\das P})=
\Val{p_1^*\ps{\das{P}}\in\eps\TO^\rho}=\ell\circ\mu^\rho
\big(\ps{\das{P}}\big)\label{dpinT=lnurho}
\end{equation}
It is easy to check that this equation actually applies to any
clopen sub-object of $\Sig$, not just those of the form
$\ps{\das{P}}$. Thus \eq{Def:CommDiagramQuantum} is indeed a
commutative diagram.

The equation \eq{dpinT=lnurho} is the precise statement of the
relationship between the topos truth value,
$\Val{p_1^*\ps{\das{P}}\in\eps\TO^\rho}$, in $\Om^{\Sh{\VH_A\times (0,1)_L}}$ and the
measure $\mu^\rho:\Subcl\Sig\map\Ga\UI$.

If we consider a normal quantum state $\rho$ (\ie a state that corresponds to a density matrix, that is, a convex combination of vector states), then we obtain a measure $\mu^\rho:\Subcl{\Sig}\map\Ga\UI$  that is $\sigma$-additive in the following sense. Let $(\ps S_i)_{i\in\bbN}$ be a countable, increasing family of clopen subobjects, then
\begin{equation}
                        \mu^\rho(\bigvee_i \ps S_i)=\bigvee_i \mu^\rho(\ps S_i).
\end{equation}
This is equivalent to Corollary IV.2 in \cite{ADmeas}. Compare this formula with \eq{Eq_sigmaAddAsJoinPreservation}: the corresponding result for classical, $\sigma$-additive measures.

Together with the fact that $\ell$ preserves joins, as proven in subsection \ref{SubSec_Qell}, this implies that $\ell\circ\mu^\rho:\Subcl{\Sig}\map\Ga\Om^{\Sh{\VH_A\times (0,1)_L}}$ preserves countable joins of increasing families $(\ps S_i)_{i\in\bbN}$ of clopen subobjects. Since, by the commutativity of the diagram \eq{Def:CommDiagramQuantum}, $\ell\circ\mu^\rho=\xi^\rho\circ p_1^*$, we have
\begin{equation}
                        (\xi^\rho\circ p_1^*)(\bigvee_i \ps S_i)=(\ell\circ\mu^\rho)(\bigvee_i \ps S_i)=\bigvee_i (\ell\circ\mu^\rho)(\ps S_i).
\end{equation}
This is the logical reformulation of normality (\ie $\sigma$-additivity) of the quantum state $\rho$. 

\subsection{The Born rule and its logical reformulation}
The diagram \eq{Def:CommDiagramQuantum} encodes a logical version of the Born rule. Propositions ``$\Ain\De$'' are represented by clopen subobjects in the topos approach. Specifically, the proposition ``$\Ain\De$'' is represented by the clopen subobject $\ps{\delta(\hat E[\Ain\De])}\in\Subcl{\Sig}$. The measure $\mu^\rho$ maps this clopen subobject to a global element of $\UI$,
\begin{equation}
                        \mu^\rho(\ps{\delta(\hat E[\Ain\De])})\in\Ga\UI.
\end{equation}
For each $V\in\VH$, we have $\mu^\rho(\ps{\delta(\hat E[\Ain\De])})(V)\in [0,1]$. As was shown in \cite{ADmeas} (and as can be seen from \eq{Def:murho}), the smallest value $\mu^\rho(\ps{\delta(\hat E[\Ain\De])})(V)$ of this global element is nothing but the usual expectation value of the projection $\hat E[\Ain\De]$ in the state $\rho$, that is
\begin{equation}                        \label{Eq_MinGlobalSec}
                        \tr(\rho\hat E[\Ain\De])=\min_{V\in\VH}\mu^\rho(\ps{\delta(\hat E[\Ain\De])})(V).
\end{equation}
In this sense, the topos approach reproduces the Born rule (for projections) and rephrases it in terms of probability measures on the spectral presheaf $\Sig$.\footnote{This can be extended to expectation values of arbitrary self-adjoint operators, as was shown in \cite{ADmeas}.} In standard quantum theory, the expectation value $\tr(\rho\hat E[\Ain\De])$ is interpreted as the probability that upon measurement of the physical quantity $A$ in the state $\rho$, the measurement outcome will lie in the Borel set $\De\subseteq\mathbb R$. This is an instrumentalist interpretation that crucially depends on concepts of measurement and observation.

We can now map the global element $\mu^\rho(\ps{\delta(\hat E[\Ain\De])})$ of $\UI$ by $\ell$ to a truth value in the sheaf topos $\Sh{\VH_A\times (0,1)_L}$. As observed earlier, this map is injective, which means that we faithfully capture the information contained in $\mu^\rho(\ps{\delta(\hat E[\Ain\De])})$ by the truth value $\ell(\mu^\rho(\ps{\delta(\hat E[\Ain\De])}))$.

As we saw in \eq{Eq_MinGlobalSec}, the expectation value $\tr(\rho\hat E[\Ain\De])$, obtained from the Born rule, is only the minimal value of the global element $\mu^\rho(\ps{\delta(\hat E[\Ain\De])})$, hence this global element and its image under $\ell$, the truth value in the sheaf topos, contain more information than just the expectation value. In future work, we will consider the question what this extra information means physically. For now, we just emphasise that the truth value
\begin{equation}
                        \ell(\mu^\rho(\ps{\delta(\hat E[\Ain\De])}))
\end{equation}
of the proposition ``$\Ain\De$'' in the state $\rho$ can be interpreted in a realist manner: the proposition refers to the physical world, the state is the state of the system (and not a description of an ensemble, or a state of knowledge, or similar), and the truth value represents one aspect of \emph{how things are}, independent of measurements and observers. Moreover, \emph{all} propositions have truth values in all states.

\section{Conclusions}
Over the years there have been a number of attempts to understand quantum theory with the aid of some sort of multi-valued logic, usually related some way to probabilistic ideas. The earliest attempt was by the Polish mathematician Lukasiewicz \cite{Lukasiewicz1913} after whose work a number of philosophers of science have made a variety of proposals. 

Much of this work involved three-valued logic, the most famous proponent of which was probably Reichenbach \cite{Reichenbach1944}. One problem commonly faced by such schemes is an uncertainty of how to define the logical connectives.  

Lukasiewicz logic that is infinite-valued has also been studied, and this includes logics whose truth values lie in $[0,1]$. This leads naturally to the subject of \emph{fuzzy logic} and \emph{fuzzy set theory}. We refer the reader to the very useful review by Pykacz \cite{Pykacz} which introduces the historical background to these ideas.  Another attempt to relate probability to logic is that of Carnap whose work, like much else, has largely disappeared into the hazy past \cite{Carnap} . 

Our topos approach is different to any of the existing schemes and, we would claim, is better motivated and underpinned with very powerful mathematical machinery. We have no problem defining logical connectives as this structure is given \emph{uniquely} by  the theory.  That is, the  propositional logic is given by the Heyting algebra $\Subcl\Sig$, and the possible truth values belong to the Heyting algebra $\Ga\Om$.

In the topos theory, truth values are not only  multi-valued, they are also \emph{contextual} in a way that is deeply tied to the underlying quantum theory. That explains why our probability measures are not simply $[0,1]$-valued but are associated with arrows from $\Subcl\Sig$ to (global elements of) the presheaf $\UI$.

In the present paper we make the strong claim that standard, \emph{classical} probability theory can be faithfully represented by the global elements of the sub-object classifier, $\Om^{(0,1)}$, in the topos, $\ShL{(0,1)}$, of sheaves on the topological space $(0,1)_L$, whose open sets correspond bijectively to probabilities in the interval $[0,1]$. This tight link between classical probability measures and Heyting algebras is captured precisely in the commutative map diagram in \eq{Def:CommDiagramClassical}. This approach to probability theory allows for a new type of non-instrumentalist interpretation that  might be particularly appropriate in `propensity' schemes.

What, to us, is rather striking is that the same can be said about the interpretation of probability in quantum theory. The relevant commutative diagram here is \eq{Def:CommDiagramQuantum} which shows how our existing topos quantum theory, which uses the topos $\Sh{\VH_A}$, can be combined with our suggested topos approach to probability, which uses the topos $\ShL{(0,1)}$, to give a new topos quantum scheme which involves sheaves in the topos $\Sh{\VH_A\times(0,1)_L}$. 

As we have seen, in this scheme, results  of quantum theory be coded in either (i) the topos probability measures $\mu^\rho:\Subcl\Sig\map\UI$ (which reinforces our slogan ``Quantum physics is equivalent to classical physics in the appropriate topos''); or (ii) the topos truth values, $\TVal{\Ain\De}{\rho}$, defined in \eq{TVinVHx}, which take their values in the Heyting algebra $\Ga\Om^{\Sh{\VH_A\times (0,1)_L}}$. 

Thus, in both classical and quantum physics, probability can be faithfully interpreted using truth values in sheaf topoi with an intuitionistic logic.

\section{Appendix}
We prove some technical results here.

First, we recall the definition in \eq{Def:aV} of the isomorphism $\a_V:\PV=\G_V \map \mathcal{C}l(\Sig_V)$ between the projections in $V$ and the clopen subsets of the
Gel'fand spectrum, $\Sig_V$, of $V$. 
A basic property of
this assignment is the existence of the commutative squares:
\begin{equation} \label{Def:CommDiagramCP1}
\begin{diagram}
        \mathcal{C}l(\Sig_V) & & \lTo^{\a_V} & \PV \\
        \dTo^{\Sig(i_{V'V})} &  & & \dTo_{\G(i_{V'V})=
                                \de(\cdot)_{V'}}  \\
        \mathcal{C}l(\Sig_{V'}) & & \lTo_{\a_{V'}} & \mathcal{P}(V')
\end{diagram}
\end{equation}
and
\begin{equation}
\begin{diagram}                 \label{Def:CommDiagramCP2}
        \mathcal{C}l(\Sig_V) & & \rTo^{\a_V^{-1}} & \PV \\
        \dTo^{\Sig(i_{V'V})} & & & \dTo_{\G(i_{V'V})=
                                \de(\cdot)_{V'}} \\
        \mathcal{C}l(\Sig_{V'}) & & \rTo_{\a_{V'}^{-1}} & \mathcal{P}(V').
\end{diagram}
\end{equation}

\begin{proposition}\label{Theorem:HypO=Subcl}
There is a bijection $k:\HG\map\Subcl\Sig$ defined for all
$\ga\in\HG$ by
\begin{equation}
        k(\ga)_V:=\a_V(\hat\ga_V)=S_{\hat\ga_V}\label{Def:kA}
\end{equation}
for all $V\in\VH$.
\end{proposition}
\begin{proof}
First, let $\ga\in\HG$. In order for $k(\hat\ga)$ to be a (clopen)
sub-object of $\Sig$, it is necessary and sufficient that, for all
$V',V\in\VH$ with $i_{V'V}:V'\subseteq V$, we have
\begin{equation}
        \Sig(i_{V'V}) \big(k(\ga)_V)\subseteq k(\ga)_{V'}
        \label{Sig(kga)A}
\end{equation}
However, the commutative square in \eq{Def:CommDiagramCP1} gives
\begin{equation}
        \Sig(i_{V'V})({S}_{\hat\a})={S}_{\G(i_{V'V})(\hat\a)}
        =S_{\das{\a}_{V'}}\\                            \label{SOSigS}
\end{equation}
for all $\hat\a\in\G_V$ and for all  $V,V'\in\VH$ with
$i_{V'V}:V'\subseteq V$. Because $\ga$ is a hyper-element of $\G$
we have $\das{\ga_V}_{V'}\preceq \hat\ga_{V'}$, and hence
\begin{eqnarray}
 \Sig(i_{V'V}) \big(k(\ga)_V\big)&=&\Sig(i_{V'V}) \big(S_{\hat\ga_V}\big)\\
 &=&S_{\das{\ga_V}_{V'}}\subseteq S_{\hat\ga_{V'}}=k(\ga)_{V'}
\end{eqnarray}
which proves \eq{Sig(kga)A}, as required. Thus the map $k$ in
\eq{Def:kA} defines a  clopen sub-object of $\Sig$.

Conversely, define a map $j:\Subcl{\Sig}\map\HG$ by
\begin{equation}
j(\ps{S})_V:=\a_V^{-1}(\ps{S}_V) =\P_{\ps{S}_V}\label{Def:jA}
\end{equation}
for all $\ps{S}\in\Subcl\Sig$ and for all $V\in\VH$. To check that
the right hand side of \eq{Def:jA} is indeed a hyper-element of
$\G$, first note that the commutative square
\eq{Def:CommDiagramCP2} gives
\begin{equation}
      \de(j(\ps{S})_V)_{V'}=  \de(\P_{\ps{S}_V})_{V'}=
        \a_{V'}^{-1}\big(\Sig(i_{V'V})(\ps{S}_V)\big)\label{doohdarA}
\end{equation}
However, since $\ps{S}$ is a sub-object of $\Sig$, we have
$\Sig(i_{V'V})(\ps{S}_V)\subseteq\ps{S}_{V'}$, and so, from
\eq{doohdarA}, for all $V'\subseteq V$,
\begin{equation}
\de(j(\ps{S})_V)_{V'}\preceq\a_{V'}^{-1}(\ps{S}_{V'})=
\P_{\ps{S}_{V'}}=j(\ps{S})_{V'}
\end{equation}
Thus the association $V\mapsto j(\ps{S})_V$ does indeed define a
hyper-element of $\G$.

It is straightforward to show that the maps $k:\HG\map\Subcl\Sig$
and $j:\Subcl{\Sig}\map\HG$ are inverses of each. Thus the theorem
is proved.
\end{proof}

\begin{proposition}\label{Theorem:SubO=HypO}
        There is a bijective correspondence
\begin{equation}
                c:\Sub\G\map\HG
\end{equation}
defined by
\begin{equation}
c(\ps{A})_V:=\bigvee\{\hat\a\mid\hat\a\in\ps{A}_V\}\label{Def:cA}
\end{equation}
for all sub-objects $\ps{A}$ of $\G$.
\end{proposition}

\begin{proof}
The first step is to show that the right hand side of \eq{Def:cA}
is a hyper-element of $\G$. To this end we note that, for any
$V'\subseteq V$,
\begin{equation}
        \de(c(\ps{A}_V))_{V'}=\de\big(\bigvee\{\hat\a\in\ps{A}_V\}\big)_{V'}=
        \bigvee\{\das\a_{V'}\mid\hat\a\in\ps{A}_V\}
\end{equation}
where we have used the fact that daseinisation commutes with the
logical $\lor$-operation. It is then clear that, since
$\{\das\a_{V'}\mid\hat\a\in\ps{A}_V\}\subseteq
\{\hat\b\mid\hat\b\in\ps{A}_{V'}\}$,
\begin{equation}
        \de(c(\ps{A}_V))_{V'}= \bigvee\{\das\a_{V'}\mid\hat\a\in\ps{A}_V\}
        \preceq\bigvee\{\hat\b\mid\hat\b\in\ps{A}_{V'}\}
        =c(\ps{A})_{V'}
\end{equation}
Therefore, the function $V\mapsto c(\ps{A})_V$ defines a
hyper-element of $\G$, as required.

Conversely, define a function $d:\HG\map\Sub\G$ by
\begin{equation}
        d(\ga)_V:=\{\hat\a\in\PV\mid\hat\a\preceq\hat\ga_V\}
        \label{Def:dA}
\end{equation}
for all $V\in\VH$. To show that the right hand side of \eq{Def:dA}
is a sub-object of the outer presheaf, $\G$, we first note that,
for all $V'\subseteq V$,
\begin{equation}
\G(i_{V'V})\big(d(\ga)_V\big)=
\de\big(\{\hat\a\in\PV\mid\hat\a\preceq\hat\ga_V\}\big)_{V'}
=\{\das\a_{V'}\mid\hat\a\preceq\hat\ga_V\}
\end{equation}
Now, if $\hat\a\preceq\hat\ga_V$ in $\PV$ then
$\das\a_{V'}\preceq\de(\hat\ga_V)_{V'}$ in $\mathcal{P}(V')$ and,
because $\ga$ is a hyper-element,
$\de(\hat\ga_V)_{V'}\preceq\hat\ga_{V'}$. Therefore,
\begin{eqnarray}
\G(i_{V'V})\big(d(\ga)_V\big)&=
&\{\das\a_{V'}\mid\hat\a\preceq\hat\ga_V\}\nonumber\\
&\subseteq&\{\das\a_{V'}\mid\das\a_{V'}
                \preceq\de(\hat\ga_V)_{V'}\}\nonumber\\
&\subseteq&\{\das\a_{V'}\mid\das\a_{V'}
        \preceq\hat\ga_{V'}\}\nonumber\\
&\subseteq&
\{\hat\b\in\mathcal{P}(V')\mid\hat\b\preceq\hat\ga_{V'}\}
=d(\ga)_{V'}
\end{eqnarray}
Thus $d(\ga)$ is a sub-object of $\G$, as claimed. It is easy to
check that $c:\Sub\G\map\HG$ and $d:\HG\map\Sub\G$ are inverses.
\end{proof}

\end{document}